\documentclass[11pt,letterpaper]{article}
\usepackage[letterpaper,margin=1in]{geometry}
\usepackage{graphicx}

\usepackage{mathtools}
\usepackage{amsthm}
\usepackage{url}
\usepackage{amsmath}
\usepackage{hyperref}
\usepackage{xcolor}
\usepackage{enumitem}
\setlist{nosep}
\usepackage{amssymb}
\usepackage{tikz}

\newtheorem{lemma}{Lemma}
\newtheorem{definition}[lemma]{Definition}
\newtheorem{theorem}[lemma]{Theorem}

\newtheorem{cor}[lemma]{Corollary}

\newcommand{\alg}{\ensuremath{\mathcal{A}}}
\newcommand{\mrs}{\text{prev}}

\newcommand{\nf}{\text{notFirst}}
\newcommand{\sts}[1]{|\triangle(#1)|}

\title{Dynamically optimal weight-balanced binary search trees}
\author{John Iacono\thanks{Université libre de Bruxelles. Research supported by the F.R.S.-FNRS.} \and Akash Pareek\thanks{Ashoka University. Research supported by the F.R.S.-FNRS while the author was at the Université libre de Bruxelles.}}
\date{}

\begin{document}

\maketitle

\begin{abstract}
In STOC 1972, J{\"{u}}rg Nievergelt and Edward M.~Reingold introduced the $BB[\alpha]$-tree, also known as a \emph{weight-balanced} tree, which for any $\alpha \in (0,\frac{1}{2}]$ is a binary search tree where every node with induced subtree size $s$ has children with induced subtree sizes at least $\alpha(s+1)-1$.
A new rotation-based binary search tree, the \emph{PopTree}, is introduced. The PopTree stores $n$ elements and is a $BB[\tfrac1{20}]$-tree after the execution of each search and a $BB[\tfrac1{729}]$-tree at all times, notably while rebalancing. For any
constant $\alpha$, its performance on search sequences of length
$\Omega(n \log n)$ is proved to be within a constant factor (which can depend
on $\alpha$) of the optimal offline rotation-based binary search tree that
maintains the $BB[\alpha]$ balance invariant.
\end{abstract}

\section{Introduction}

The dynamic optimality conjecture for general rotation-based binary search trees (BST)~\cite{DBLP:journals/jacm/SleatorT85} posits that there exists an online rotation-based BST algorithm whose performance is within a constant factor of the optimal offline BST algorithm for every sufficiently long sequence; i.e., there exists a $O(1)$-competitive rotation-based BST.
This conjecture has received much attention with numerous partial results, for example~\cite{DBLP:journals/siamcomp/Wilber89,DBLP:conf/soda/WangDS06,DBLP:journals/siamcomp/DemaineHIP07,DBLP:conf/birthday/Iacono13,DBLP:conf/esa/Munro00,DBLP:conf/soda/DemaineHIKP09,DBLP:conf/wads/Fox11,DBLP:conf/icalp/DemaineILO13,DBLP:conf/focs/ChalermsookG0MS15,DBLP:conf/esa/ChalermsookG0MS15,DBLP:conf/wads/ChalermsookG0MS15,DBLP:conf/isaac/ChalermsookG0MS18,DBLP:conf/approx/ChalermsookCS20,DBLP:conf/soda/LevyT19,DBLP:conf/soda/IaconoL16,DBLP:conf/soda/ChalermsookGJAP23,DBLP:journals/corr/abs-2607-18498,DBLP:conf/icalp/Chalermsook0JPY24}.
In addition to the non-constant competitive ratio results~\cite{DBLP:conf/soda/WangDS06,DBLP:journals/siamcomp/DemaineHIP07,DBLP:conf/soda/LevyT19,DBLP:journals/corr/abs-2607-18498},
the failure to obtain dynamic optimality has resulted in the study of adaptive, or beyond-worst-case analysis of general binary search trees to determine which kind of patterns in search sequences can be taken advantage of, and which can not, which has given rise to the study of closed-form bounds including
the working set bound~\cite{DBLP:journals/jacm/SleatorT85,DBLP:conf/icalp/DemaineILO13},
the sequential access bound~\cite{DBLP:conf/wads/Fox11,DBLP:conf/wads/ChalermsookG0MS15},
the dynamic finger bound and weighted dynamic finger bound~\cite{DBLP:conf/soda/IaconoL16},
the unified bound~\cite{DBLP:conf/birthday/Iacono13,DBLP:conf/icalp/DemaineILO13},
and bounds for pattern-avoiding sequences~\cite{DBLP:conf/focs/ChalermsookG0MS15,DBLP:conf/isaac/ChalermsookG0MS18,DBLP:conf/soda/ChalermsookGJAP23,DBLP:conf/icalp/Chalermsook0JPY24}.
This procession of bounds gives some the hope that improving these explicit bounds might lead to a tight bound matching the lower bounds, which so far have been limited to the decades-old bounds of Wilber~\cite{DBLP:journals/siamcomp/Wilber89,DBLP:conf/approx/ChalermsookCS20} and the perhaps-equivalent independent rectangle bounds~\cite{DBLP:conf/soda/DemaineHIKP09}.

However, such rich and intensive study is almost entirely absent from balanced binary search trees, which given their ubiquitousness and the undoubted practical importance of maintaining a logarithmic worst-case running time, is surprising.

This work is meant to be a starting point in answering questions of the form: ``Given binary search trees obeying some balance invariant, can you characterize the best rotation-based BSTs that maintain that invariant from the perspective of competitive ratio, dynamic optimality, or beyond worst case bounds?''
We emphasize that our view here is to maintain the invariants of balanced BSTs while ignoring the classical update methods, instead allowing any rotation-based implementation of the operations that maintains the balance invariant. In this paper, we look at the class $BB[\alpha]$ trees, as balance is their explicit invariant. In doing so we develop a closed-form expression with matching upper and lower bounds which completely characterizes the asymptotic dynamically optimal running time of  $BB[\alpha]$-trees. Unsurprisingly, our dynamic optimality characterization shows that there are many classes of sequences for which balanced BSTs are slower than optimal general BSTs, but this is of course useful information telling you when staying in the class of $BB[\alpha]$-trees is sufficient and when for performance reasons you must leave it. What is more surprising is the explicit bound for dynamically optimal $BB[\alpha]$-trees is not one of the extant BST bounds but rather something new.

\paragraph{The $BB[\alpha]$-tree}
The $BB[\alpha]$-tree is a type of weight-balanced binary search tree (BST) introduced by J{\"{u}}rg Nievergelt and Edward M.~Reingold \cite{DBLP:journals/siamcomp/NievergeltR73}. We begin with their definition:
\begin{definition}
\label{d:balanced}
\textbf{Definition of a $BB[\alpha]$-tree.}
Let $\Delta(x)$ be the size of the subtree induced at node $x$ in a BST and define the weight of a node $x$, $w(x)$, to be $\Delta(x)+1$.
Let $\alpha$ be a constant in the range $(0,\frac{1}{2}]$ and let $\beta \coloneqq 1-\alpha$.
A $BB[\alpha]$-Tree is a binary search tree which for all nodes $x$ with parent $y$: $$\alpha w(y)
\leq w(x) \leq \beta w(y).$$
A missing child has a subtree size of zero and is thus treated as having weight one.
\end{definition}

The reason for using weight instead of subtree size is that even small examples, such as a two-node BST, would not meet such a definition. However, with weights, a two-node tree is a $BB[\frac{1}{3}]$-tree since the root has weight 3 and the missing child of the root has weight 1.
It was shown that no BSTs are $BB[\alpha]$-trees for $\frac{1}{3} < \alpha\leq \frac{1}{2}$ except perfectly-balanced BSTs which are $BB[\frac{1}{2}]$-trees~\cite{DBLP:journals/siamcomp/NievergeltR73}.
One nice side effect of weight is that the weight of a node is the sum of the weights of its children.
The class of $BB[\alpha]$ trees does have its limitations. It is easy to see that a length-$m$ repeated sequential access on a tree with $n$ items takes time $\Theta(m \log n)$ on any $BB[\alpha]$ tree, while on splay trees they take $O(m)$ time \cite{DBLP:journals/combinatorica/Tarjan85}.

We completely characterize the asymptotic running time of the offline optimal $BB[\alpha]$ trees with rotation up to constant factors for sufficiently long sequences of searches. We present a new BST, the \emph{PopTree}, which is a $BB[\tfrac1{729}]$-tree whose performance is proved to be within a constant factor of the optimal offline rotation-based $BB[\alpha]$-tree for any constant $\alpha\in(0,\frac{1}{2}]$ and a sufficiently long sequence of searches.

\paragraph{The model.}
Before presenting this characterization, we define precisely what we mean by a rotation-based $BB[\alpha]$ algorithm, which we will henceforth call a $BB[\alpha]$-algorithm, following the standard formalization of rotation-based BSTs with the addition of the $BB[\alpha]$ balance invariant.

As we do not consider insertions or deletions, we assume that the BST stores the integers from 1 to $n$.
The $BB[\alpha]$ algorithm executes a sequence of searches $X=\langle x_1, x_2, \ldots x_m \rangle$, each of which is an item of $[n]$.
The BST algorithm has a single pointer, which at unit cost begins each search at the root and at unit cost can move the pointer to the parent, left child, or right child, or it can perform a single rotation involving the node pointed to and the parent of the node pointed to.
Any computational time or data stored to determine which unit-cost computations are to be performed are not counted, and an algorithm in the model is free to be offline and make its choices as any function of the entire search sequence.
To execute the search for $x_i$, the algorithm performs any sequence of unit cost operations provided that (i) the pointer at some point touches the node containing $x_i$
and (ii) the tree is a valid $BB[\alpha]$ at all times during the execution.
Each BST algorithm may require an initial starting tree or a class of allowable starting trees.

\paragraph{The optimal running time.}
Let the set $\nf$ be all $i$ such that $x_i$ is not the first search of that item.
Let the \emph{previous search} at time $i$ and range $r$, $\mrs_i(r)$, be defined for any $i \in \nf$ to be $$\mrs_i(r) \coloneqq\max\{j \in [i-1]| x_j \in [x_i-r,x_i+r]  \}.$$
Intuitively, this is the time of the most recent search before $x_i$ with an item within a range $\pm r$ from $x_i$.
As $x_i$ was searched before, $\mrs_i(r)$ is well-defined as it will return the time of the previous search to $x_i$ if no other item within the range $\pm r$ from $x_i$ was searched between the time of $x_i$'s previous search and $i$.

We now define our optimal running time for a search sequence $X$:
\begin{equation} \label{e:opt} OPT(X)
\coloneqq
\sum_{i \in \nf}
\left(
\log \min\{r \geq 1|
x_{\mrs_i(\frac{n}{r})} = x_i\}
+
\log
\frac{n}{\min(x_i,n-x_i+1)}
\right)
\end{equation}

\paragraph{Lower bound.} Our main lower bound result shows that any $BB[\alpha]$ algorithm, including offline algorithms, has a running time of $\Omega(OPT(X))$, where the constant in the $\Omega$ depends on $\alpha$. Our upper bound is the presentation of the PopTree data structure, which is an online $BB[\frac{1}{729}]$-algorithm and whose running time is $O(OPT(X))$ for sequences of length $\Omega(n \log n)$. Together, they show that \eqref{e:opt} is the complete asymptotic characterization of the running time of optimal $BB[\alpha]$ algorithms, and that our PopTree data structure is $O(1)$-competitive with the optimal offline algorithm, otherwise known as being dynamically optimal.

We provide some intuition for \eqref{e:opt}. One simple observation is that not every item in a $BB[\alpha]$ tree can be the root, which implies that the closer an item is to the minimum or maximum, the deeper it must be in the tree. This is formalized in Lemma~\ref{l:extreme} and is responsible for the second term in \eqref{e:opt}.
For the first term, observe that in a $BB[\alpha]$-tree, given two items with consecutive key values, one of them must be at logarithmic depth. By generalizing this observation to items whose keys are $r$ apart, one gets that one of them must be at depth $\Omega(\log \frac{n}{r})$ (Lemma~\ref{l:dun}). So, to execute a search on some item $x$ in time $o(\log \frac{n}{r})$ will be problematic if there was some item $x'$ with $|x-x'|\leq r$ searched recently. Getting the accounting to work for this case is tricky and requires a charging argument. Such an argument is needed because at the time $x$ is searched, $x$ could be shallow and thus cheap to search, and $x'$ depth would be $\Omega(\log \frac{n}{r})$, which superficially is not a problem but gives hints of contours of the charging scheme we develop.

Our $OPT$ formula is new and, in some ways, the inverse of the bounds known for general rotation-based BSTs.
For such BSTs, the bounds usually capture the concepts of spatial locality (sequential access~\cite{DBLP:journals/combinatorica/Tarjan85}, dynamic finger~\cite{DBLP:journals/siamcomp/ColeMSS00,DBLP:journals/siamcomp/Cole00}) and temporal locality (working set bound~\cite{DBLP:journals/jacm/SleatorT85}) and were combined in the unified bound~\cite{DBLP:conf/soda/Iacono01a,DBLP:journals/tcs/BadoiuCDI07}. The unified bound says that a search is fast if it is close (in key-space) to something that was searched recently. Compare this to our $OPT$ bound, which intuitively says that a search is fast if nothing close (in key-space) was searched recently (since the previous search of the item). It is fascinating that arbitrary BSTs reward locality in key-space while weight-balanced BSTs punish it!

\paragraph{Upper bound. } Our PopTree data structure is presented first as a type of skip list~\cite{DBLP:journals/cacm/Pugh90}, where each level has invariants
that imply the minimum and maximum number of keys that can be spanned between items. This contrasts with the typical skip list's concern with bounding between two items, and with the number of items one level down.
We also have carefully crafted invariants that result in recently searched items being closer to the top of the list if possible, but where the converse does not hold, thus allowing a certain flexibility in restructuring.
We promote (pop!) the searched item to the top level of the skip list and then do a series of promotions and demotions to restore the invariants of the structure.
We convert the skip list to a BST structure via the intermediate step of a B-tree.
See Figure~\ref{fig:skippoplist-to-bst} for an illustration of a SkipPopTree and the conversion process to a BST via the intermediate step of a B-tree.
This conversion process is largely standard, and our invariants are designed to show easily that the resulting BST is a $BB[\tfrac1{729}]$-tree at all times and a $BB[\tfrac1{20}]$-tree after the completion of any operation. The $\alpha=\frac{1}{729}$ is only needed during transient states in the middle of an execution, but it serves to show that the $BB[\alpha]$ invariant can be maintained at all times, which matches what we require in our lower bound.
\begin{figure*}[p]
  \centering
  \resizebox{\textwidth}{!}{
}\\[4pt]
  {\small (a)~The SkipPopList on $n=3^4=81$ keys. The buckets are defined based on powers of three. Extreme buckets must be empty, occupied buckets cannot be adjacent, and three vacant buckets in a row are prohibited.}\\[14pt]

  \resizebox{\textwidth}{!}{%
}\\[4pt]
  {\small (b)~The PopBTree obtained from (a): each maximal run of consecutive
   top items of a level forms one node.}\\[14pt]

  \resizebox{\textwidth}{!}{%
}\\[4pt]
  {\small (c)~The PopTree (BST): the right-chain expansion of (b), with each
   former node shown as a shaded oval.}

  \caption{From a SkipPopList to a weight-balanced BST, on one instance.}
  \label{fig:skippoplist-to-bst}
\end{figure*}

\paragraph{Key-independent optimality. }
In~\cite{DBLP:journals/algorithmica/Iacono05}, Iacono introduced the notion of key-independent optimality. Given a search sequence $X$ and a random bijection $f:[n]\mapsto [n]$, let $f(X)=\langle f(x_1), f(x_2), \ldots, f(x_m) \rangle$. The intuitive idea is that if we do not expect any correlations in our data based on proximate key values, for example, if key values are some sort of randomly assigned ID, the distributional properties of $X$ and $f(X)$ should not differ. In~\cite{DBLP:journals/algorithmica/Iacono05} it was shown that the expected value of executing $f(X)$ on the offline optimal BST with rotations was asymptotically the same as what is known as the working set property. As splay trees were shown to have the working set property~\cite{DBLP:journals/jacm/SleatorT85}, this showed the existence of a BST that, given any sequence $f(X)$, executes the sequence in time asymptotically the same as the best offline optimal. Such structures were called key-independently optimal.
We show in Section~\ref{s:ws} that in expectation $OPT(f(X))$ is
asymptotically equal to the working-set bound.
Thus, we show our PopTree is key-independently optimal and more significantly show that $BB[\alpha]$ trees, while much more restricted than general rotation-based BSTs, are all that is needed in the key-independently optimal setting.

\subsection{Related work}

\paragraph{Weight-balanced trees.}
The $BB[\alpha]$-tree, or weight-balanced tree, was introduced by Nievergelt
and Reingold \cite{DBLP:journals/siamcomp/NievergeltR73}, who proposed maintaining balance by
subtree weight rather than by height, achieving logarithmic search and update
times using single and double rotations. Most subsequent work on the
structure has concerned not new algorithmic capabilities but the correct
choice of the balance parameter and of the rebalancing rules. Blum and
Mehlhorn \cite{DBLP:journals/tcs/BlumM80} identified an error in the original
analysis, established the admissible parameter range, and proved that the
amortized number of rebalancing operations per update is constant. Lai and
Wood \cite{DBLP:journals/ijfcs/LaiW93} gave a top-down rebalancing procedure.
The definitive treatment of the parameter question was given by Hirai and
Yamamoto \cite{DBLP:journals/jfp/HiraiY11}, who precisely characterized the
valid parameter space and formally verified the rotation conditions in the
Coq proof assistant, showing that the rules as originally stated fail for
some parameters previously believed valid. Weight-balanced trees found
renewed relevance in functional programming through the variant of Adams
\cite{DBLP:journals/jfp/Adams93}, which underlies the standard persistent set
and map libraries of several functional languages. Recently Jug{\'e}
\cite{DBLP:conf/wads/Juge25} introduced a relaxation in which the weight
condition is imposed between nodes and their grandchildren. In all of this
work, the object of study is the maintenance of the invariant itself. The
question we ask here: the best running time achievable by \emph{any}
sequence of rotations preserving the invariant appears not to have been
studied.

\paragraph{Dynamic optimality for binary search trees.}
Sleator and Tarjan \cite{DBLP:journals/jacm/SleatorT85} introduced splay
trees and conjectured that they are $O(1)$-competitive with the optimal
offline rotation-based BST. This \emph{dynamic optimality conjecture} remains
open, the best known competitive ratio for any BST is the $O(\log \log n)$ of
tango trees \cite{DBLP:journals/siamcomp/DemaineHIP07}. See
\cite{DBLP:conf/birthday/Iacono13} for a survey. Our result can be read as a
resolution of the dynamic optimality question for the class of BST algorithms
that maintain the $BB[\alpha]$ invariant.

\paragraph{Dynamically optimal skip lists and B-trees.}
The work closest in spirit to ours is that of Bose, Dou{\"i}eb, and
Langerman \cite{DBLP:conf/soda/BoseDL08} on dynamic optimality for skip
lists and B-trees. They showed that for skip lists satisfying a weak
balancing property, the working-set bound is a lower bound on the cost of
executing any search sequence. They constructed a deterministic
self-adjusting skip list, which promotes each searched item to the topmost
level and then restores the balancing property, whose cost matches the
bound. The structure is therefore dynamically optimal within its class, with
consequences for $(a,b)$-trees, B-trees, and their BST encodings via
standard dualities \cite{DBLP:conf/ACMse/DeanJ07}. Their result and ours
exhibit the same phenomenon: while dynamic optimality for unrestricted
rotation-based BSTs remains open, imposing a balance invariant on the class
of structures collapses the offline optimum to an explicit formula that a
concrete online structure can match. Indeed, our PopTree takes direct
inspiration from their promotion-based structures: it too pops the searched
item to the top level of a skip-list representation and repairs its
invariants with promotions and demotions.
The two characterizations,
however, are strikingly different. For balanced skip lists, the optimum is
the working-set bound, which is well known from the analysis of splay trees and which rewards temporal locality, recently searched
items are cheap. For $BB[\alpha]$-trees, the optimum instead
punishes proximity in key space. A search is cheap only if nothing nearby
was searched recently, and includes a term forcing items near the minimum
or maximum key to be expensive.

We further note that in the \emph{unrestricted} B-tree model, where no
balance condition is imposed, only $O(\log \log n)$-competitiveness is known
\cite{DBLP:journals/mst/DemaineIKL21}, reinforcing that it is the balance
restriction that makes full dynamic optimality tractable.

\paragraph{Distribution-sensitive balanced structures.}
A parallel line of work asks not for the offline optimum but for specific
distribution-sensitive bounds achievable \emph{while} maintaining balance.
Bent, Sleator, and Tarjan \cite{DBLP:journals/siamcomp/BentST85} introduced
biased search trees, a weight-balanced structure whose access time is
$O(1+\log(W/w_i))$ for an item of weight $w_i$ in a tree of total weight $W$.
Bagchi, Buchsbaum, and Goodrich \cite{DBLP:journals/algorithmica/BagchiBG05}
gave a skip-list counterpart, including a deterministic worst-case variant.
Closer to the self-adjusting setting, Bose, Dou{\"i}eb, Dujmovi{\'c}, and Howat
\cite{DBLP:journals/algorithmica/BoseDDH12} constructed a height-balanced BST
that attains the working-set bound in the worst case and the unified bound
amortized, and Martel \cite{DBLP:journals/ipl/Martel91}, and Sherk
\cite{DBLP:journals/jal/Sherk95} studied self-adjusting multiway and $k$-ary
trees. In each case, a balance invariant is retained, and a fixed runtime goal
(weighted, working-set, or unified) is matched. None characterizes, as we do,
the offline optimum over all rotation sequences preserving the invariant.

\paragraph{Organization of the paper.}
In Section~\ref{s:lower} we prove the lower bound: every
$BB[\alpha]$-algorithm, online or offline, requires time $\Omega(OPT(X))$
to execute any search sequence $X$, where the constant depends on
$\alpha$. Section~\ref{s:upper} contains the matching upper bound. We
first present the SkipPopList, a skip list with occupancy invariants that
tie recently searched items to the top levels, and show that it executes
any search sequence of length $\Omega(n \log n)$ in time $O(OPT(X))$ in
the skip list cost model. Section~\ref{s:popbtree} transforms the
SkipPopList into weight-balanced multiway search tree, the PopBTree, and
Section~\ref{s:poptree} transforms the PopBTree into our final structure,
the PopTree, a BST maintained by single rotations that is a
$BB[\tfrac1{20}]$-tree after every search, a $BB[\tfrac1{729}]$-tree
after every unit-cost operation, and runs in time $O(OPT(X))$.

Section~\ref{s:ws} we show that in expectation $OPT$ of a search sequence
under a random bijection of the keys equals the working-set bound,
establishing that PopTrees are key-independently optimal.

We note that the two transformations are placed in the appendix to reflect that these are basically standard transformations, albeit ones where the balance factor needs to be carefully propagated.
Our SkipPopList was designed to have these transformations be intuitive.
As noted in the acknowledgments, this cleanly separates the paper into the entirely human conceived and written main body of the paper, and the machine written but human edited and verified grind of the two transformations in the appendices which we feel most readers would find believable after looking at Figure~\ref{fig:skippoplist-to-bst}.
\section{Lower Bound}
\label{s:lower}

This section is about $BB[\alpha]$-trees that store the integers in $[n]$ and are executing a sequence of $m$ searches $X=\langle x_1, x_2, \ldots x_m \rangle$ to items in $[n]$.

\subsection{Depth lemmas}

We begin by proving some easy lemmas about the depth of items in a $BB[\alpha]$-tree.
The \emph{depth} of a node $x$ in a BST, $d(x)$, is the length of the path from the root to $x$. By repeatedly applying the balance condition in Definition~\ref{d:balanced} we obtain:

\begin{lemma}
\label{l:sts}
\textbf{Weights and depths.}
$$
\log_{\frac{1}{\alpha}} \frac{n+1}{w(x)}
\leq d(x) \leq
\log_{\frac{1}{\beta}} \frac{n+1}{w(x)} $$
The same bounds hold when $x$ is a missing child treated as a weight-one node, per Definition~\ref{d:balanced}.
\end{lemma}

\begin{lemma}\label{l:extreme}\textbf{Extreme depths.}
The depth of $x$ is at least $\log_{\frac{1}{\alpha}} \frac{n}{\min(x,n-x+1)}-1$.
\end{lemma}

\begin{proof}
Node $x$'s left child $x_L$ (possibly a missing child treated as a weight-one node) has depth $d(x_L)=d(x)+1$ and subtree size $\Delta(x_L)\leq x-1$ and weight $w(x_L)\leq x$.
Node $x$'s right child $x_R$ has depth $d(x_R)=d(x)+1$ and subtree size $\Delta(x_R)\leq n-x$ and thus weight $w(x_R)\leq n-x+1$.
Apply Lemma~\ref{l:sts}, using these weights to get the bounds on the depths of $x_L$ and $x_R$ (using $\frac{n+1}{w} \geq \frac{n}{w}$) and subtract one to get the claimed bound on the depth of $x$.
\end{proof}

\begin{lemma}
\label{l:dac}
\textbf{Depth ancestor/child.}
Given any node $x$ with ancestor  $z$:
$$ d(x) \geq \log_{\frac{1}{\alpha}} \frac{n}{|x-z|}-1.$$
\end{lemma}

\begin{proof}
Assume w.l.o.g.~$x<z$. Let $y$ be the right child of $x$ (possibly a missing child treated as a weight-one node). We know $\sts{y} \leq  |x-z|-1$ and thus $w(y)\leq |x-z|$. We then apply Lemma~\ref{l:sts} to $y$ and subtract 1 to get the bound on $x$.

\end{proof}

\begin{lemma} \label{l:dun}
\textbf{Close nodes imply one is deep.}
Given any different nodes $x$ and $z$:
$$\max(d(x),d(z)) \geq \log_{\frac{1}{\alpha}} \frac{n}{|x-z|}- 1.$$
\end{lemma}

\begin{proof}
If $x$ and $z$ have an ancestor/descendant relationship then apply Lemma~\ref{l:dac}.
Otherwise let $y$ be the LCA of $x$ and $z$. We know that $|x-y|+|y-z|=|x-z|$. We then apply Lemma~\ref{l:dac} to $y$ and $x$,  and $y$ and $z$ which shows
$d(x) \geq \log_{\frac{1}{\alpha}} \frac{n}{|x-y|}- 1$
and
$d(z) \geq \log_{\frac{1}{\alpha}} \frac{n}{|z-y|}- 1$,
both of which imply the lemma since $|x-y|<|x-z|$ and $|z-y|<|x-z|$.
\end{proof}

\subsection{Levels}

Given a $BB[\alpha]$ tree $T$, the \emph{level} of a node $\ell(x)$ in the tree is defined to be $\lceil \log_{\frac{1}{\alpha}} \frac{n+1}{w(x)} \rceil$. As the root has weight $n+1$, level is a nonnegative integer with the root having level exactly 0.
The main feature of levels is that while they largely behave like depths up to constant factors, when performing rotations only the nodes involved can change level while the rest remain the same. This property does not hold for depths.

From the definition of a $BB[\alpha]$ tree we now obtain the following:

\begin{lemma} \label{l:childlevel}
\textbf{Level of a child.}
Given node $x$ with child $y$,
either $\ell(x)=\ell(y)$ or $\ell(x)+1=\ell(y)$.
\end{lemma}

\begin{proof}
By Definition~\ref{d:balanced}, $\alpha w(x) \leq w(y) \leq \beta w(x) < w(x)$.
From $w(y) < w(x)$ we get $\log_{\frac{1}{\alpha}} \frac{n+1}{w(y)} > \log_{\frac{1}{\alpha}} \frac{n+1}{w(x)}$, and as the ceiling is nondecreasing, $\ell(y) \geq \ell(x)$.
From $w(y) \geq \alpha w(x)$ we get
$\log_{\frac{1}{\alpha}} \frac{n+1}{w(y)} \leq \log_{\frac{1}{\alpha}} \frac{n+1}{w(x)} + 1$, and as $\lceil t+1 \rceil = \lceil t \rceil + 1$, $\ell(y) \leq \ell(x)+1$.
\end{proof}

We can also show that depths and levels only differ by a constant.

\begin{lemma} \label{l:ldiff}
\textbf{Depths and levels.}
For any node $x$, $\ell(x) -1  \leq d(x) \leq \frac{\ell(x)}{\log_\alpha \beta}$.

\end{lemma}

\begin{proof} In the following, the first step is from Lemma~\ref{l:sts}, followed by math and the definition of $\ell(x)$.
$$ d(x) \geq \log_{\frac{1}{\alpha}} \frac{n+1}{w(x)}
> \left\lceil \log_{\frac{1}{\alpha}} \frac{n+1}{w(x)} \right\rceil -1 = \ell(x)-1 $$

$$ d(x)
\leq
\log_{\frac{1}{\beta}} \frac{n+1}{w(x)}
=
\frac{1}{\log_{\frac{1}{\alpha}}\frac{1}{\beta}}
\log_{\frac{1}{\alpha}} \frac{n+1}{w(x)}
\leq
\frac{1}{\log_{\frac{1}{\alpha}}\frac{1}{\beta}}
\left\lceil  \log_{\frac{1}{\alpha}} \frac{n+1}{w(x)}   \right\rceil
= \frac{\ell(x)}{\log_\alpha \beta}$$

\end{proof}

We thus take the depth result of Lemma~\ref{l:dun} and use the $\ell(x) \geq \log_\alpha \beta \cdot d(x)$ of Lemma~\ref{l:ldiff} to give a level-based result.

\begin{lemma} \textbf{Close nodes imply high level. } \label{l:onedeep}
Given any different nodes $x$ and $z$, one has level at least $ \log_\alpha \beta \log_{\frac{1}{\alpha}} \frac{n}{|x-z|}-1$.
\end{lemma}

\begin{proof}
    \begin{align*}
\max(\ell(x),\ell(z))
&\geq
\max(\log_\alpha \beta \cdot d(x),
\log_\alpha \beta\cdot d(z))
& \text{Lemma~\ref{l:ldiff}}
\\
&=
\log_\alpha \beta \max(  d(x),
 d(z))
\\
&\geq
\log_\alpha \beta \left(\log_{\frac{1}{\alpha}} \frac{n}{|x-z|}- 1 \right)
& \text{Lemma~\ref{l:dun}}
\\
&\geq \log_\alpha \beta  \log_{\frac{1}{\alpha}}  \frac{n} {|x-z|} -1
& 0< \log_\alpha \beta \leq 1
\end{align*}
\end{proof}

\begin{lemma}
\label{c:deeplevel}
\textbf{$\gamma(\ell)$: the threshold of closeness to be at level $\ell$. }
Let $\gamma(\ell)\coloneqq n
\alpha^{\frac{\ell+1}{\log_\alpha \beta} }$.
Given any two nodes $x$ and $z$, if $|x-z| \leq  \gamma(\ell)$, the level of at least one of $x$ and $z$ is at least $\ell$.
\end{lemma}
\begin{proof}
The inequality
$\max(\ell(x),\ell(z))
 \geq \log_\alpha \beta  \log_{\frac{1}{\alpha}}  \frac{n} {|x-z|}-1$
from Lemma~\ref{l:onedeep}
and
$|x-z| \geq \gamma(\max(\ell(x),\ell(z)))$ are equivalent via algebra. Thus we have
$$ \gamma(\ell) \geq |x-z| \geq \gamma(\max(\ell(x),\ell(z))).$$

This implies that $\ell \leq \max(\ell(x),\ell(z))$ since $\gamma$ is decreasing in $\ell$.
\end{proof}

\subsection{Algorithm}

In this section, we assume an algorithm $\alg$ which executes the search sequence $X=\langle x_1,x_2, \ldots x_m \rangle $ where each search is an item in $[n]$. Let $T_i$ denote the tree before the execution of $x_i$.

An algorithm begins each search by initializing the pointer to be the root at unit cost and does a series of unit cost operations by moving the pointer left, right, up, or doing a single rotation with the parent with the single rule that at some point in the execution of the $i$th search, $x_i$, the pointer must touch $x_i$. Additionally, to be a valid $BB[\alpha]$-algorithm, the tree must be a valid $BB[\alpha]$-tree at all times.

Let $A(\alg,T_0,X)$ be the cost of the algorithm to execute $X$ starting from the initial tree $T_0$, in terms of the unit cost operations previously presented.

All other notation from this point on in this section may implicitly depend on $\alg$, $X$, and $T_0$.
Other notation previously defined with reference to a BST may now be subscripted to refer to $T_i$, e.g.~$\ell_i(x_i)$ is the level of the $i$th item to be searched just before the algorithm commences. For brevity, we refer to the node $p$ points to, as well as the nodes themselves, by the value of the data in the node. As each node stores a unique value in $[n]$, this is unambiguous.

We now introduce notation to give a unique identifier to each unit-cost operation.
Note that it is not sufficient to simply identify the items the pointer points to during the course of executing a search, as a pointer could point to an item several times.

We define the set $A$ to be all 4-tuples $(i,j,p,\ell)$ where $p$ is the value stored in the node pointed to by the pointer after the $(j+1)$st unit-cost operation during the execution of the $i$th search and $\ell$ is the level of $p$ in the tree at that moment. Note that $i$ is in the range $[1..m]$ as there are $m$ searches, and $j$ is in the range from 0 to one less than the number of unit-cost operations used to execute the $i$th search. The value $j=0$ will represent the state after the first unit-cost operation of the $i$th search, the initialization of the pointer to the root of the tree.
Thus, $A$ provides a full transcript of the execution of the algorithm, and also encodes the level of the pointer at every step.

\begin{lemma} \textbf{Relating the size of $A$ and the running time of the algorithm.} \label{l:asize}
$|A|=A(\alg,T_0,X)$.
\end{lemma}
\begin{proof}
    The elements of $A$ are in bijection with the unit-cost operations performed by the algorithm: $(i,j,p,\ell)$ corresponds to the $(j+1)$st unit-cost operation of the $i$th search.
\end{proof}

\subsection{Levels and algorithm}

\begin{lemma} \label{l:levelchange}
\textbf{Only rotations change levels. }
In a rotation, only the levels of the two nodes involved can change. Thus, the level of a node can not change during the execution of a search before the first time the pointer touches the item.
\end{lemma}
\begin{proof}
The level is a function of weight, which changes only at the two nodes involved in a rotation. For the second claim, we show that at all times during a search, every ancestor of the pointed node has been touched during that search. The claim follows as the two nodes involved in a rotation are the pointed node $x$ (touched by definition) and its parent $y$ (an ancestor). The invariant holds by induction on the unit-cost operations, it holds after the initialization of the pointer to the root, which has no ancestors. Moving the pointer to a child adds only the previously pointed node to the set of ancestors of the pointed node. Moving the pointer to the parent removes it from the set, and a rotation of the pointed node $x$ with its parent $y$ removes $y$ from the set.
\end{proof}

\begin{lemma} \label{l:contpath}
\textbf{The pointer touches nodes at all levels smaller than the level of what it is searching. }
Given any $0 \leq \ell' \leq  \ell_i(x_i)$, there must be a $p$ and $j$ such that there is a $(i,j,p,\ell') \in A$.
\end{lemma}

\begin{proof}
Lemma~\ref{l:childlevel}, showing a child is at the same level or +1 compared to the parent almost is the proof but a bit more detail is needed to make sure that as the pointer performs rotations that this can not be used to skip a level.
The tuple $(i,0,\text{root},0)$ covers the case $\ell'=0$. Each unit-cost operation changes the level of the pointed node by at most one: moving to a child changes it by $0$ or $+1$ (Lemma~\ref{l:childlevel}), moving to the parent by $0$ or $-1$, and in a rotation the pointed node takes the position, and thus the weight and level, previously held by its parent, a change of $0$ or $-1$ (Lemma~\ref{l:childlevel}). By Lemma~\ref{l:levelchange}, the first time the pointer touches $x_i$ during search $i$, its level is still $\ell_i(x_i)$. The sequence of levels in the tuples $(i,0,\cdot,\cdot), (i,1,\cdot,\cdot), \ldots$ therefore starts at $0$, at some point reaches $\ell_i(x_i)$, and changes by at most one per step, and so takes every value $\ell' \in [0,\ell_i(x_i)]$.
\end{proof}

\begin{lemma}
\label{l:conttime}
\textbf{Downward continuity of the level of an item. }
Given, times $i$ and  $i''$ where $i<i''$, item $y$, and level $\ell$ with $\ell_{i''}(y) < \ell \leq \ell_i(y)$. There is a time $i'\in[i,i'')$ and a $j$ such that $(i',j,y,\ell)\in A$.

\end{lemma}

\begin{proof}
By Lemma~\ref{l:levelchange}, the level of $y$ changes only in rotations involving $y$, and any such change is by exactly one: consider a rotation of a node $x$ with parent $y'$, and set $\ell'=\ell(y')$ before the rotation. Before the rotation, $\ell(x)$ was $\ell'$ or $\ell'+1$ and after it will be $\ell'$. Before the rotation $\ell(y')$ is $\ell'$ and after it will be $\ell'$ or $\ell'+1$ (Lemma~\ref{l:childlevel}). Thus, when $y$ is the parent in a rotation, its level stays the same or increases by one. Its level \emph{decreases} only when $y$ is the pointed node and is rotated up with its parent, and in that case, the pointer is at $y$ immediately before the rotation, so a tuple $(\cdot,\cdot,y,\ell(y))$ with $y$'s pre-rotation level is in $A$.
The level of $y$ is at least $\ell$ at time $i$ and less than $\ell$ at time $i''$ and moves in unit steps, so during some search $i' \in [i,i'')$ it decreases from $\ell$ to $\ell-1$. Thus, the touch of $y$ at level $\ell$ immediately before that rotation is the required tuple.
\end{proof}

\subsection{Previous search}

Recall that $\nf$ is the subset of $[m]$ consisting of $i$ where $x_i$ is not the first search of $x_i$ in $X$.
As we will often be looking at the previous search of an item, it will be helpful to exclude items being searched for the first time from our analysis.
Observe $|\nf|\geq m-n$ as each item in $[n]$ can only be searched for the first time once.

Given an $i \in \nf$, define $\mrs_i$ to be the time of the previous search to $x_i$ before time $i$:
$$
\mrs_i \coloneqq
\max\{i'| i'<i\text{ and }x_{i'}=x_i\}.$$

Given an $i \in \nf$, define $\mrs_i(r)$ to be the $j$ such that $x_j$ is the most recently searched item in
$[x_i-r,x_i+r]$
before time $i$:
$$
\mrs_i(r) \coloneqq
\max \{i' | i'<i\text{ and }x_{i'}\in  [x_i-r,x_i+r]\}.
$$
This is well-defined as $x_i$ has been searched before and this previous search ensures that the set will have at least one element, $\mrs_i$. This observation also gives the following lemma:

\begin{lemma} \label{l:recent}
For any $i \in \nf$, and any $r \geq 0$,  $\mrs_i \leq \mrs_i(r)$
\end{lemma}

\begin{proof}
    Observe that for fixed $i$, $\mrs_i(r)$ is nondecreasing as a function of $r$. Furthermore for $r=0$, $\mrs_i(r)=\mrs_i$.
\end{proof}

\subsection{Marking scheme}

For every $i\in \nf$ and every nonnegative integer $\ell$ where $x_i \not = x_{\mrs_i\left( \gamma(\ell) \right)}$ we define the following marking scheme which maps $(i,\ell)$ to a touch $(i',j,p,\ell)$ in $A$.

\begin{description}
\item[Type 1 mark.] If $\ell_i(x_i) \geq  \ell$
then by Lemma~\ref{l:contpath}
there is $j$ and $p$ so that $(i,j,p,\ell) \in A$. We mark this touch with a type-1 mark. We use the minimum $j$ if there are multiple possibilities. Intuitively, when searching for something deeper than level $\ell$, one must touch a node at level $\ell$ on the way.

\item[Type 2 marks.] In this case we assume a type 1 mark is not possible
($\ell_i(x_i) <  \ell$) and \\
$\ell_{\mrs_i(\gamma(\ell))}(x_{\mrs_i(\gamma(\ell))}) \geq  \ell$.
By Lemma~\ref{l:contpath}
there is $j$ and $p$ so that $(\mrs_i(\gamma(\ell)),j,p,\ell) \in A$. We use the minimum $j$ if there are multiple possibilities.
Intuitively, this is the same logic as the type 1 mark, except in this case the deep search is not $x_i$ but rather $x_{\mrs_i(\gamma(\ell))}$.

\item[Type 3 mark.] If type 1 and type 2 marks are not possible, we have $\ell_i(x_i) < \ell$ and $\ell_{\mrs_i(\gamma(\ell))}(x_{\mrs_i(\gamma(\ell))})< \ell$.
Observe that $|x_i - x_{\mrs_i(\gamma(\ell))}|\leq \gamma(\ell)$ by the definition of $\mrs$.
We now apply Lemma~\ref{c:deeplevel} to the tree $T_{\mrs_i(\gamma(\ell))}$ and at least one of $x_i$ and $x_{\mrs_i(\gamma(\ell))}$ must be at level at least $\ell$ in this tree. Since $\ell_{\mrs_i(\gamma(\ell))}(x_{\mrs_i(\gamma(\ell))})< \ell$ we conclude that it must be $x_{i}$ as $\ell_{\mrs_i(\gamma(\ell))}(x_{i})\geq \ell$.

Now we have at time $i$ the level of $x_{i}$ was less than $\ell$ and at time $\mrs_i(\gamma(\ell))$ the level of $x_{i}$ was at least $\ell$, i.e.,
$\ell_{\mrs_i(\gamma(\ell))}(x_{i})\geq \ell$ and
$\ell_i(x_{i}) < \ell$.
Thus by the downward continuity of levels, Lemma~\ref{l:conttime} applied with times $\mrs_i(\gamma(\ell))$ and $i$, $x_{i}$ must have been touched at level $\ell$ during the execution of some search $i'$ with $\mrs_i(\gamma(\ell))\leq i' < i$. It is the first such touch $(i',j,x_{i},\ell)\in A$ that we mark with a type 3 mark.
\end{description}

\begin{lemma} \label{l:fourmarks}
No touch in $A$ receives more than four marks.
\end{lemma}

\begin{proof}
Specifically, we show that no touch in $A$ receives more than one Type-1 mark, two Type-2 marks and one Type-3 mark.
\begin{description}
\item[Type 1:]
Given $i$ and $\ell$, this marks $(i,j,p,\ell)$ for some $p$ and $j$ and thus a different choice of $i$ and/or $\ell$ will type-1 mark a different element of $A$.

\item[Type 2 marks:] We show that no touch receives more than two type-2
marks.

A type-2 mark generated by the pair $(i,\ell)$ marks a touch of the form
$(\mrs_i(\gamma(\ell)),j,p,\ell)$. In particular the fourth coordinate of
the marked touch is $\ell$. Thus pairs with different values of $\ell$ mark
different touches, and we fix $\ell$ for the remainder of the argument. For
brevity, let $r\coloneqq\gamma(\ell)$, and for a pair $(i,\ell)$ generating a
type-2 mark, let $q_i\coloneqq\mrs_i(r)$ be the time of the marked touch.
Two pairs $(i,\ell)$ and $(\hat{i},\ell)$ can mark the same touch only if
$q_i=q_{\hat{i}}$. Therefore, it suffices to show that at most two indices $i$
generating type-2 marks share the same value of $q_i$.

Suppose $(i,\ell)$ and $(\hat{i},\ell)$ both generate type-2 marks, with
$i\neq\hat{i}$ and $q_i=q_{\hat{i}}\eqqcolon q$, without loss of generality assume
$i<\hat{i}$. By the definition of $\mrs$, $q<i<\hat{i}$. We establish three
facts.
\begin{itemize}

    \item \emph{Fact 1: $|x_i-x_q|\leq r$ and $|x_{\hat{i}}-x_q|\leq r$.} This is
immediate from the definition of $\mrs_i(r)$ and $\mrs_{\hat{i}}(r)$: the
item searched at time $q$ has key in $[x_i-r,x_i+r]$ and in
$[x_{\hat{i}}-r,x_{\hat{i}}+r]$.

\item \emph{Fact 2: $x_i\neq x_q$ and $x_{\hat{i}}\neq x_q$.} The marking scheme
is only defined for pairs $(i,\ell)$ with $x_i\neq x_{\mrs_i(\gamma(\ell))}$.

\item \emph{Fact 3: $|x_i-x_{\hat{i}}|>r$.} Suppose not, so
$x_i\in[x_{\hat{i}}-r,x_{\hat{i}}+r]$. Since $q<i<\hat{i}$, time $i$ is then
a candidate in the maximum defining $\mrs_{\hat{i}}(r)$, giving
$q=\mrs_{\hat{i}}(r)\geq i$. This contradicts $q<i$.

\end{itemize}

By Fact~2, neither $x_i$ nor $x_{\hat{i}}$ equals $x_q$, leaving two cases
for their positions relative to $x_q$.

\emph{Case 1: $x_i$ and $x_{\hat{i}}$ lie on the same side of $x_q$.}
Assume both are smaller than $x_q$, the case where both are larger is
symmetric. By Fact~1 both lie in $[x_q-r,x_q)$, an interval of length $r$,
so $|x_i-x_{\hat{i}}|<r$, contradicting Fact~3. This case is impossible.

\emph{Case 2: $x_i$ and $x_{\hat{i}}$ lie on opposite sides of $x_q$.} This
is the only remaining possibility.

Therefore, among all pairs $(i,\ell)$ generating type-2 marks with $q_i=q$,
at most one has $x_i<x_q$ and at most one has $x_i>x_q$, since two on the
same side would realize the impossible Case~1. Thus at most two pairs mark
touches at time $q$ with level $\ell$, and no touch receives more than two
type-2 marks.

\item[Type 3 marks:] Given $i$ and $\ell$, this always marks $(i',j,x_i,\ell)$ for some $i'$ in the range $[\mrs_i(\gamma(\ell)),i)$.  By Lemma~\ref{l:recent} $[\mrs_i(\gamma(\ell)),i) \subseteq [\mrs_i,i)$. Observe that if $x_i=x_{\hat{i}}$, $i\not=\hat{i}$ that the intervals $[\mrs_i,i)$ and $[\mrs_{\hat{i}},\hat{i})$ are disjoint.
Thus no other search $x_{\hat{i}}$ and $\ell$ could mark $(i',j,x_i,\ell)$ as the touch it would mark would not have an $i'$-value in $[\mrs_i,i)$.
\end{description}

\end{proof}

\begin{lemma} \label{l:prefinal} The time to execute $X$ on a $BB[\alpha]$-tree is at least
$$
\frac{1}{4} |\{ (i,\ell) | i \in \nf \text{ and }  x_{\mrs_i(\gamma(\ell))} \not = x_i \}|.
$$
\end{lemma}

\begin{proof}
Each element of $\{ (i,\ell) | i \in \nf \text{ and }  x_{\mrs_i(\gamma(\ell))} \not = x_i \}$ marks an element of $A$. By Lemma~\ref{l:fourmarks}, each element of $A$ receives at most 4 marks, so at least
$$\frac{1}{4} |\{ (i,\ell) | i \in \nf \text{ and }  x_{\mrs_i(\gamma(\ell))} \not = x_i \}|$$ distinct elements of $A$ receive marks, giving a lower bound on the size of $A$. Finally by Lemma~\ref{l:asize}
$|A|=A(\alg,T_0,X)$, so a lower bound on $|A|$ is a lower bound on the execution time of the algorithm.
\end{proof}

\begin{theorem} \label{t:lb}
The time to execute $X$ on a $BB[\alpha]$-tree with constant $\alpha$ is $\Omega(OPT(X))$:
$$
\Omega\left(
\sum_{i \in \nf}
\left(
\log \min\left\{r\geq 1 \,\middle|\,
x_{\mrs_i(\frac{n}{r})} = x_i\right\}
+
\log
\frac{n}{\min(x_i,n-x_i+1)}
\right)
\right).
$$
\end{theorem}

\begin{proof}
\textbf{The second term:} From Lemma~\ref{l:extreme},
the depth of $x_i$ in the tree $T_i$ is at least $\log_{\frac{1}{\alpha}} \frac{n}{\min(x_i,n-x_i+1)}-1$.
Each unit-cost operation changes the distance in the current tree between the pointed node and $x_i$ by at most one. This is immediate for the pointer movements, and in a rotation of the pointed node $x$ with its parent $y$, the distance from $x$ increases by one for nodes in the subtree of the child of $x$ that becomes a child of $y$, decreases by one for nodes outside the subtree formerly rooted at $y$, and is unchanged otherwise. As search $i$ begins with the pointer at the root, at distance $d_{T_i}(x_i)$ from $x_i$, the pointer must touch $x_i$ and search $i$ using at least $d_{T_i}(x_i)$ unit-cost operations. Thus
$$A(\alg,T_0,X) \geq \sum_{i \in \nf} \log_{\frac{1}{\alpha}} \frac{n}{\min(x_i,n-x_i+1)} - m.$$

\textbf{The first term:} This is a manipulation of Lemma~\ref{l:prefinal}:
\begin{align*}
A(\alg,T_0,X) &\geq \frac{1}{4} |\{ (i,\ell) \mid i \in \nf \text{ and }  x_{\mrs_i(\gamma(\ell))} \not = x_i \}|
\\
&= \frac{1}{4} \sum_{i \in \nf} |\{ \ell \mid   x_{\mrs_i(\gamma(\ell))} \not = x_i \}|.
\\
\intertext{Observe that if $x_{\mrs_i(\gamma(\ell))} \not = x_i$ then $\mrs_i(\gamma(\ell)) > \mrs_i$ (Lemma~\ref{l:recent}), and thus the same holds for all smaller nonnegative integer values of $\ell$. Note that $\gamma$ is decreasing in $\ell$ and $\mrs_i(r)$ is nondecreasing in $r$, so $\mrs_i(\gamma(\ell')) \geq \mrs_i(\gamma(\ell)) > \mrs_i$ for all $\ell' \leq \ell$, and no search at a time in $(\mrs_i,i)$ is a search of $x_i$. The set $\{ \ell \mid x_{\mrs_i(\gamma(\ell))} \not = x_i \}$ is therefore downward closed, and its size is precisely the smallest $\ell$ such that $x_{\mrs_i(\gamma(\ell))} = x_i$, which we denote $k_i$. This is well defined because for large enough $\ell$, $\gamma(\ell)<1$ and the only integer in $[x_i-\gamma(\ell),x_i+\gamma(\ell)]$ is $x_i$ itself, from which $\mrs_i(\gamma(\ell))=\mrs_i$.}
&= \frac{1}{4} \sum_{i \in \nf} k_i
\end{align*}
It remains to relate $k_i$ to the minimization over $r$ in the statement of the theorem. Let
$$r_i \coloneqq \frac{n}{\gamma(k_i)} = \left(\frac{1}{\alpha}\right)^{\frac{k_i+1}{\log_\alpha \beta}}.$$
Then $r_i \geq 1$ and $x_{\mrs_i(n/r_i)} = x_{\mrs_i(\gamma(k_i))} = x_i$, so $r_i$ is a candidate in the minimization and
$$\min\left\{r\geq 1 \,\middle|\, x_{\mrs_i\left(\frac{n}{r}\right)} = x_i\right\} \leq r_i = \left(\frac{1}{\alpha}\right)^{\frac{k_i+1}{\log_\alpha \beta}}.$$
Taking $\log_{\frac{1}{\alpha}}$ of both sides and rearranging gives
$$ k_i \geq \log_\alpha \beta \cdot \log_{\frac{1}{\alpha}} \min\left\{r\geq 1 \,\middle|\, x_{\mrs_i\left(\frac{n}{r}\right)} = x_i\right\} - 1.$$
Summing over $i \in \nf$ and using $|\nf| \leq m$,
$$A(\alg,T_0,X) \geq \frac{1}{4} \sum_{i \in \nf} k_i \geq \frac{\log_\alpha \beta}{4}\, S_1 - \frac{m}{4},
\qquad \text{where } S_1 \coloneqq \sum_{i \in \nf}
\log_{\frac{1}{\alpha}}
\min\left\{r \geq 1 \,\middle|\,
x_{\mrs_i\left(\frac{n}{r}\right)}
= x_i
\right\}.
$$

\textbf{Combining:} Let $S_2 \coloneqq \sum_{i \in \nf} \log_{\frac{1}{\alpha}} \frac{n}{\min(x_i,n-x_i+1)}$, so that the second-term bound reads $A(\alg,T_0,X) \geq S_2 - m$.
Both $S_1$ and $S_2$ are nonnegative. Each term of $S_1$ is at least $\log_{\frac{1}{\alpha}} 1 = 0$ by the restriction $r \geq 1$, and each term of $S_2$ is nonnegative as $\min(x_i,n-x_i+1) \leq n$.
Trivially $A(\alg,T_0,X) \geq m$, as each search uses at least one unit-cost operation.
Taking $\frac{4}{5}$ of the first-term bound plus $\frac{1}{5}$ of the trivial bound gives
$$ A(\alg,T_0,X) \geq \frac{4}{5}\left( \frac{\log_\alpha \beta}{4}\, S_1 - \frac{m}{4}\right) + \frac{1}{5}\, m = \frac{\log_\alpha \beta}{5}\, S_1,$$
and taking half of the second-term bound plus half of the trivial bound gives $A \geq \frac{1}{2} S_2$. Averaging these two bounds,
$$ A(\alg,T_0,X) \geq \frac{\log_\alpha \beta}{10}\, S_1 + \frac{1}{4}\, S_2 \geq \frac{\log_\alpha \beta}{10} \left(S_1 + S_2\right).$$
As $\alpha$ is constant, $\log_{\frac{1}{\alpha}}$ and $\log$ agree up to a constant factor, and thus $A(\alg,T_0,X) = \Omega(OPT(X))$.
\end{proof}
\section{Upper Bound}
\label{s:upper}

In this section, we begin by defining the SkipPopList, which is a variant of the skip-list data structure of Pugh~\cite{DBLP:journals/cacm/Pugh90}. After proving SkipPopList has a running time of $O(OPT(X))$, in the following subsections we will show how via a detailed transform drawn from \cite{DBLP:conf/ACMse/DeanJ07} to our BST algorithm, the PopTree, a $BB[\frac{1}{729}]$-algorithm with a running time of $O(OPT(X))$ can be obtained. This transformation does not proceed directly but rather via a multiway search tree, which we call the PopBTree.
The SkipPopList is the most interesting part, as the transformations are largely standard but require somewhat tedious arguments to ensure that constant balance is preserved every step of the way.
In the transform from SkipPopList to the PopBTree, one nontrivial detail is that in a SkipPopList, one can move freely left and right in a level. However, in the PopBTree, only parent and child pointers are present. Hence, an argument that one can effectuate the specific changes of the SkipPopList without blowing up the running time needs to be presented.

\subsection{SkipPopList}

The structure is defined over a static set of data of size $n$, which, for convenience, we assume to be a power of three and which contains the integers in $[n]$.

\paragraph{Skip list.} We begin with the usual ingredients of a Skip List:
\begin{enumerate}
    \item There are levels, which we number from the top down, starting at level 1. Note that the concept of level used when discussing skip lists is different from that in the previous section, which applies to trees.
    \item Each level stores a subset of the data.
    \item Level $\ell$ stores a subset of the data in level $\ell+1$.
    \item There is a level $\ell_{\max}$ which stores $[n]$, all the data.
    \item The level of an item, $\ell(x)$ is the smallest level that it appears in.
\end{enumerate}

\paragraph{Skip list cost model. }
We make the skip list cost model precise. There is a single pointer that always points to an item in a particular level and starts each search at the smallest-valued item in level 1. At unit cost, the pointer can move to the predecessor or the successor on the current level, or stay with the same item and move up or down a level, so long as the item is present in the destination level. Additionally, if the pointer points to the top most level an item belongs to, it may \emph{demote} the item which involves removing the item from that level and moving the pointer down one level. Items in $\ell_{\max}$ can not be demoted. The inverse operation of \emph{promotion} is also allowed, which inserts the item into the level above and moves the pointer up, and is permitted only when the pointer is at the item's topmost level, and it is not already at level~1. A skip list successfully executes a search for $x_i$ if the pointer touches a node containing $x_i$ at any level during the execution of the search.

\paragraph{SkipPopList Invariants. }
We now describe the basic structural invariants of a SkipPopList.
\begin{enumerate}[resume]
    \item In level $\ell$ we view all keys (all of them, not just the ones stored at that level) as being partitioned into $3^\ell$ \emph{buckets}. The $3^\ell$ buckets each contain $n/3^\ell$ consecutive keys. We adopt the tree terminology of sibling, parent, and child. In general, a bucket has two sibling buckets,
    a parent bucket (above) and three child buckets (below).
    \item Each bucket at level $\ell$ is either disjoint from the items stored at level
    $\ell$, which we refer to as \emph{vacant}, or the bucket intersects a single item stored at level $
    \ell$ in which case we refer to it as \emph{occupied}.
    \item No adjacent buckets on the same level can be occupied
    \item \label{i:empty} At most two consecutive buckets on the same level can be vacant.
    \item The minimum and maximum buckets on every level must be vacant, these are referred to as the \emph{extreme} buckets.
\end{enumerate}
See Figure~\ref{fig:skippoplist-to-bst} for an illustration.
The SkipPopList structure also has temporal invariants with respect to a sequence of searches being performed on items in $[n]$.
These invariants are the most novel part of our construction. Intuitively, they require recently searched items that do not have a more recently searched item nearby to be in lower-numbered levels as a function of ``nearby'' and ``recently,'' but yet give flexibility to allow rebalancing to work.

\begin{enumerate}[resume]
    \item Each nonextreme bucket is either \emph{important} or \emph{unimportant}. Extreme buckets are simply extreme and must be vacant.
    \item A nonextreme bucket is important if its most recently searched item was searched more recently than the most recently searched items in its two sibling buckets.
    \item An important bucket must be occupied by the most recently accessed item in the bucket.
    \item Unimportant buckets may be either occupied or vacant and if they are occupied there is no requirement on which item occupies it.
\end{enumerate}

\paragraph{The bottom level.}
Set $\ell_{\max} \coloneqq \log_3 n + 1$. For levels $\ell \leq \log_3 n$ the buckets are as defined above: $3^\ell$ buckets of $\frac{n}{3^\ell}$ consecutive keys each. In particular, at level $\log_3 n$ each bucket consists of a single key. At level $\ell_{\max}$ there are $3n$ buckets, three children per level-$(\log_3 n)$ bucket. The key of the parent bucket belongs to the \emph{middle} child, and the two outer children contain no keys. Level $\ell_{\max}$ stores all of $[n]$. (Assigning keys by the natural subdivision of key ranges would instead place each key in an outer child, and in particular place key $n$ in the maximum bucket of level $\ell_{\max}$, violating the requirement that extreme buckets be vacant, hence the middle-child convention.)

\paragraph{Assorted observations}

\begin{lemma}{\bf The bottom level}\label{l:bottom}
The occupancy of level $\ell_{\max}=\log_3 n+1$ is fully determined, i.e., the middle child of each level-$(\log_3 n)$ bucket is occupied by that bucket's key, and every other bucket is vacant, so exactly every third bucket is occupied. This occupancy satisfies all SkipPopList invariants and never changes during the execution of any search sequence.
\end{lemma}

\begin{proof}
Level $\ell_{\max}$ stores all of $[n]$ and each key belongs to the middle child of its parent, so the occupied buckets are exactly the $n$ middle children. The outer children contain no keys and are therefore permanently vacant. Read left to right, the level has the pattern, vacant, occupied, vacant, repeated $n$ times.

We verify the invariants. Every occupied bucket contains exactly one stored item, its key. Consecutive occupied buckets are separated by exactly two vacant buckets, so no two adjacent buckets are occupied, and no three consecutive buckets are vacant. The minimum and maximum buckets are outer children and thus vacant. For the temporal invariants, each occupied bucket contains a single item, which is therefore its most recently searched item whenever any of its items has been searched, and its sibling buckets contain no keys. So an important bucket at level $\ell_{\max}$ is always occupied by its most recently searched item.

Finally, the occupancy never changes. Promotions insert an item into a level above its topmost level and remove nothing from the levels below. A demotion at a level removes the occupant of a bucket that contains, or is adjacent to, the bucket containing the searched key at that level. At level $\ell_{\max}$, the buckets adjacent to the searched key's bucket contain no keys and are permanently vacant, and distinct keys occupy buckets at a distance of at least three. So no conflict ever arises at level $\ell_{\max}$ and no clearing cascade ever demotes an item from it.
\end{proof}

\begin{lemma} \label{l:below}
For any bucket $b$ at level $\ell < \ell_{\max}$, at least one of the three child buckets of $b$ at level $\ell+1$ is occupied. Moreover, if $b$ is vacant, the occupant of any occupied child bucket of $b$ has $\ell+1$ as its topmost level.
\end{lemma}

\begin{proof}
The three child buckets of $b$ are consecutive at level $\ell+1$ and partition the keys of $b$. By Invariant~\ref{i:empty}, they cannot all be vacant, so at least one is occupied. For the second claim, if the occupant $w$ of an occupied child bucket were stored at some level $\ell' \leq \ell$, then by the subset property, $w$ would be stored at level $\ell$. As $w$'s key lies in $b$, bucket $b$ would then be occupied, a contradiction.
\end{proof}

\begin{lemma} \label{l:deeplevel}
Let $i \in \nf$ and let $\ell$ be a level, $1 \leq \ell \leq \ell_{\max}$.
If $x_{\mrs_i(\frac{2n}{3^\ell})}=x_i$ and $\min(x_i,n-x_i+1)> \frac{n}{3^\ell}$, then $x_i$ is stored in level $\ell$ of the SkipPopList just before the execution of the $i$th search.
\end{lemma}

\begin{proof}
Since $\min(x_i,n-x_i+1)> \frac{n}{3^\ell}$ we have $x_i > \frac{n}{3^\ell}$ and $x_i \leq n-\frac{n}{3^\ell}$, so $x_i$ lies in neither the minimum nor the maximum bucket of level $\ell$. Therefore, its level-$\ell$ bucket is not extreme.

The range $[x_i-\frac{2n}{3^\ell},x_i+\frac{2n}{3^\ell}]$ includes all keys of the size-$\frac{n}{3^\ell}$ level-$\ell$ bucket containing $x_i$ as well as all keys of that bucket's two sibling buckets. As $x_{\mrs_i(\frac{2n}{3^\ell})}=x_i$ means $x_i$ is the most recently searched item in $[x_i-\frac{2n}{3^\ell},x_i+\frac{2n}{3^\ell}]$ before time $i$, it is the most recently searched item in its bucket and more recently searched than any item in the two sibling buckets. Its bucket is therefore important, and by the temporal invariants an important bucket is occupied by its most recently searched item, namely $x_i$. Hence $x_i$ is stored in level $\ell$.
\end{proof}

\begin{lemma} \label{l:singledeepopt}
Given an $i \in \nf$,
let $\ell_i(x_i)$ be the minimum level of $x_i$ in the SkipPopList before the execution of the $i$th operation. We have
$$
\ell_i(x_i)
\leq
\max
\left(
\log_3
\min \left\{ r \geq 1 \,\middle|\,
x_{\mrs_i\left( \frac{n}{r} \right)}  = x_i
\right\}
,
\log_3
\frac{n}{\min(x_i,n-x_i+1)}
\right)+3
$$
\end{lemma}

\begin{proof}
Follows from Lemma~\ref{l:deeplevel} and some math. Throughout, $\ell$ ranges over the nonnegative integers both conditions of Lemma~\ref{l:deeplevel} are upward closed in $\ell$ (shrinking the range of $\mrs_i$ preserves $x_i$ being most recent, and $\frac{n}{3^\ell}$ is decreasing), $\ell=0$ fails the second condition as $\min(x_i,n-x_i+1)\leq n$, and if the minimum below exceeds $\ell_{\max}$ the first inequality holds trivially since $\ell_i(x_i)\leq \ell_{\max}$.

\begin{align*}
\ell_i(x_i)
&\leq  \min_\ell \left(
x_{\mrs_i(\frac{2n}{3^\ell})}=x_i\text{ and }\min(x_i,n-x_i+1)> \frac{n}{3^\ell}\right)
\intertext{By upward closure, this is the maximum level where the condition fails, plus one:}
&= 1+ \max_\ell \left(
x_{\mrs_i(\frac{2n}{3^\ell})}\not =x_i\text{ or }\min(x_i,n-x_i+1)\leq  \frac{n}{3^\ell}\right)
\\
&= 1+ \max \left( \max_\ell
x_{\mrs_i(\frac{2n}{3^\ell})}\not =x_i,\ \max_\ell \left[ \ell \leq \log_3 \frac{n}{ \min(x_i,n-x_i+1)}
\right]
\right)
\\
\intertext{Since $\ell$ is an integer, $\max_\ell[\ell \leq a]=\lfloor a \rfloor \leq a$.}
& \leq  1+ \max \left( \max_\ell
x_{\mrs_i(\frac{2n}{3^\ell})}\not =x_i,\ \log_3 \frac{n}{ \min(x_i,n-x_i+1)}
\right)
\\
\intertext{The maximum integer where something fails is the minimum where it holds, minus one.}
& =  1+ \max \left(-1+\min_\ell
x_{\mrs_i(\frac{2n}{3^\ell})} =x_i,\ \log_3 \frac{n}{ \min(x_i,n-x_i+1)}
\right)
\intertext{We now allow the level variable to be a nonnegative real, denoted $\hat\ell$, to be more flexible with the math. As the condition is upward closed, the minimum over the integers exceeds the minimum over the reals by at most one, so:}
& \leq
1+\max \left(
\min\left\{\hat{\ell} \,\middle|\,
x_{\mrs_i(\frac{2n}{3^{\hat{\ell}}})} =x_i\right\},\
\log_3 \frac{n}{ \min(x_i,n-x_i+1)}
\right)
\\
\intertext{As $\hat\ell = \log_3 \frac{3^{\hat{\ell}}}{2} + \log_3 2$ and $\frac{3^{\hat{\ell}}}{2}$ is monotone increasing in $\hat\ell$:}
& =
1+\max \left(
\log_3\min\left\{ \tfrac{3^{\hat{\ell}}}{2} \,\middle|\,
x_{\mrs_i(\frac{2n}{3^{\hat{\ell}}})} =x_i\right\}+\log_3 2,\
\log_3 \frac{n}{ \min(x_i,n-x_i+1)}
\right)
\\
\intertext{Set $r=\frac{3^{\hat{\ell}}}{2}$, so $\frac{2n}{3^{\hat\ell}}=\frac{n}{r}$. As $\hat\ell$ ranges over the nonnegative reals, $r$ ranges over $[\frac12,\infty)$. Further restricting to $r\geq 1$ can only increase the minimum. Using $1+\log_3 2 < 2 < 3$:}
&\leq  \max \left(
\log_3\min\left\{r\geq 1 \,\middle|\,
x_{\mrs_i(\frac{n}{r})} =x_i\right\},\
\log_3 \frac{n}{ \min(x_i,n-x_i+1)}
\right)+3
\end{align*}
\end{proof}

Thus we can conclude that the sum of all depths of searched nodes is asymptotically the same as the lower bound of $O(OPT(X))$.

\begin{cor} \textbf{Depth sums.}\label{cor:depth}
For any search sequence $X$ of length $\Omega(n \log n)$
$$
\sum_{i} \ell_i(x_i) = O(OPT(X)).
$$
\end{cor}
\begin{proof}
First observe $OPT(X)=\Omega(m)$. Each term of \eqref{e:opt} is nonnegative, and its second summand is at least $\log \frac{n}{\min(x_i,n-x_i+1)} \geq \log\frac{2n}{n+1} \geq \log \frac{4}{3}$ for $n \geq 2$, so $OPT(X) = \Omega(|\nf|)$. As $|\nf| \geq m-n$ and $m=\Omega(n\log n)$, $OPT(X)=\Omega(m)$.

For $i \in \nf$, Lemma~\ref{l:singledeepopt} bounds $\ell_i(x_i)$ by the maximum of two nonnegative quantities plus $3$. The maximum is at most the sum, and as both arguments of the logarithms are at least $1$, each base-3 logarithm is at most the corresponding base-2 logarithm, so the sum is at most the $i$th term of \eqref{e:opt}. Thus
$\sum_{i \in \nf}\ell_i(x_i) \leq OPT(X) + 3|\nf| = O(OPT(X))$.

Also $\sum_{i \in [m]\setminus\nf} \ell_i(x_i) = O(n\log n)$ because there are at most $n$ first searches and every item is at level at most $\ell_{\max} = O(\log n)$. Since $OPT(X)=\Omega(m)$ and $m=\Omega(n\log n)$, this is also $O(OPT(X))$.
\end{proof}

\paragraph{Search algorithm.}

We have just completed the hard work of linking the depths of the searched items to $OPT$. Now we just need to show that we can perform unit cost operations to maintain the invariants of the SkipPopList in time proportional to the depth of the item being searched. Fortunately, this is straightforward.

When executing a search sequence, changes need to be made to the structure to restore any temporal invariants that are now violated. In particular, after executing $x_i$, any non-extreme bucket that $x_i$ belongs to becomes important and needs to have $x_i$ as its representative.

 Here is the method. Let $t$ be the level at which the promotion of $x_i$ will stop, the highest level whose bucket containing $x_i$ is not extreme.
\begin{itemize}
    \item \textbf{Clearing demotions.} For each item $y \neq x_i$, that at some level in $[t,\ell_i(x_i))$, occupies either the bucket containing $x_i$ or a bucket adjacent to it. We demote $y$ repeatedly, starting at $y$'s topmost level and proceeding downward, until $y$ is no longer stored at any level at which its bucket contains, or is adjacent to, the bucket containing $x_i$'s key. (The levels at which $y$ conflicts form a contiguous range ending at some level $c_y < \ell_i(x_i)$, since buckets only grow when moving up. So this means demoting $y$ until its topmost level is $c_y+1$.)
    \item \textbf{Promotion.} Promote $x_i$ one level at a time until the bucket above is extreme, or level 1 is reached.
    \item \textbf{Gap repair.} Wherever an item was demoted, if the resulting vacancy creates three consecutive vacant buckets on that level, promote into the middle of the three the occupant of one of its child buckets, which exists and has the level below as its topmost level by Lemma~\ref{l:below}. This process should proceed in a bottom-up order.
\end{itemize}

\begin{lemma}\label{l:maintain}
The above method restores all SkipPopList invariants.
\end{lemma}

\begin{proof}
We first consider the temporal invariants. After the search of $x_i$, the important buckets that change status are: every nonextreme bucket containing $x_i$ becomes important (as $x_i$ is now the most recently searched item overall), and
every bucket adjacent to a bucket containing $x_i$'s key, extreme or not,
becomes unimportant.
At levels $\geq \ell_i(x_i)$, $x_i$ was already stored and hence already occupied every bucket containing it, and no adjacent bucket was occupied, so nothing needs to change there.

At levels $[t,\ell_i(x_i))$, the promotion places $x_i$ in every newly important bucket. Above level $t$, every bucket containing $x_i$'s key is extreme (the ancestors of an extreme bucket are extreme) and thus exempt from importance.
The clearing demotions remove items only from buckets that contain, or are
adjacent to, the bucket containing $x_i$'s key at that level. A vacated
bucket adjacent to the bucket containing $x_i$'s key is unimportant after
the search of $x_i$, as a sibling now holds the most recently searched item,
and unimportant buckets may be vacant. A vacated bucket that itself contains
$x_i$'s key can occur only at a level in $[t,\ell_i(x_i))$ because above $t$ this
bucket is extreme and hence vacant, and at levels $\geq \ell_i(x_i)$ it is
occupied by $x_i$ itself. Such a bucket is important after the search, and
the subsequent promotion phase fills it with $x_i$, its most recently
searched item, restoring the invariant.

For the structural invariants, the clearing demotions ensure that when $x_i$ is promoted into a bucket, that bucket is otherwise unoccupied and its neighbors are vacant, so occupancy-uniqueness and the no-adjacent-occupied invariant hold at every level $x_i$ is promoted to. Demotions and gap-repair promotions never place two items in one bucket or occupy a bucket adjacent to an occupied one (the repaired bucket's two neighbors are vacant by construction).
Extreme buckets are never promoted, as promotion stops below them. Gap repair fills only a vacant bucket, both of whose neighbors are vacant,
which is an extreme bucket, having a single neighbor.

Finally, since the invariants held before the search, after the demotions and
the promotion of $x_i$, every maximal run of vacant buckets has length at most
four. Moreover, a run of length four can arise only at a level above $t$, where it
includes the extreme bucket containing the key of $x_i$.
Each such run is eliminated
by a gap-repair promotion into a vacant bucket, both of whose neighbors are
vacant. At most two such promotions are needed per level.
Each repaired bucket is unimportant after the search of $x_i$, so any
occupied child's occupant may be promoted into it. First, it does not
contain $x_i$'s key: at levels at least $t$ that bucket is occupied by $x_i$.
By the time gap repair runs, and repair fills only vacant buckets, while
above level $t$, it is extreme, and extreme buckets are never repaired. If a
sibling of the repaired bucket contains $x_i$'s key, the repaired bucket is
unimportant, as that sibling holds the most recently searched item.
Otherwise, its importance status is unchanged by the search, and before the
search, it was vacant, a bucket vacated by a demotion contains $x_i$'s key
or is adjacent to it, cases already excluded. Hence, unimportant, as
important buckets must be occupied.
\end{proof}

\begin{lemma} \label{l:skipcost}
In the skip list model, searching an item at level $\ell$ and restoring the invariants of the SkipPopList takes time $O(\ell)$.
\end{lemma}

\begin{proof}For the search, consecutive occupied buckets on any level are at most three
buckets apart. So two consecutive items at level $\ell'$ lie in a span of at
most four consecutive buckets of level $\ell'$, and the key range between
them is thus contained in at most twelve buckets of level $\ell'+1$, of
which at most six are occupied, as no two adjacent buckets are occupied.
Descending one level, therefore, requires at most a constant number of
predecessor/successor moves, and reaching an item at level $\ell$ takes
$O(\ell)$ time.

For the restoration, we bound the total number of unit-cost demotions and promotions. At each level, at most three buckets, the bucket containing $x_i$'s key and its two neighbors can contain an occupant that must be demoted, and each gap repair performs one promotion adjacent to a demotion site. The promotion of $x_i$ contributes one operation per level. Hence, the total number of structural operations is $O(\ell)$. Every operation site is within a constant number of buckets of the bucket containing $x_i$'s key at its level, so moving the pointer between consecutive sites costs $O(1)$ horizontal moves plus vertical moves. Since demotions themselves move the pointer down one level and promotions up one level, and the sites are processed in a single top-down pass per item, the total vertical movement is $O(\ell)$.
\end{proof}

\begin{lemma}{\bf Weak invariants during restoration. }\label{l:weakinv}
Consider the execution of the restoration method for the search of $x_i$, performed in the order: clearing demotions, then promotion of $x_i$, then gap repair. At every point during this execution (i.e., after every unit-cost operation) and at every level $\ell < \ell_{\max}$:
\begin{enumerate}
    \item no two adjacent buckets are occupied, and
    \item at most seven consecutive buckets are vacant.
\end{enumerate}
\end{lemma}

\begin{proof}
For a level $\ell$, let $B_\ell$ denote the bucket at level $\ell$ containing the key $x_i$, and let $W_\ell$ consist of $B_\ell$ together with its (at most two) sibling buckets, a window of at most three consecutive buckets. We first record the key structural fact: \emph{every demotion performed at level $\ell$ removes the occupant of a bucket in $W_\ell$.} Indeed, an item $y$ is demoted from level $\ell$ only during its top-down clearing cascade, and at every level at which $y$ is demoted, its bucket contains, or is adjacent to, the bucket containing the key $x_i$, which lies in $W_\ell$. Consequently, any bucket outside $W_\ell$ that was occupied before the search remains occupied throughout the execution, as demotions are the only operations that vacate buckets.

\emph{Claim 1.} Occupancy is created only by promotions, of which there are two kinds. When $x_i$ is promoted into $B_\ell$, the clearing phase has already vacated both siblings of $B_\ell$ (any occupant of a sibling of $B_\ell$ at a level into which $x_i$ is promoted is, by definition, cleared), and no promotion refills them before the gap-repair phase. Hence, at the moment of promotion, both neighbors of $B_\ell$ are vacant. A gap-repair promotion fills the middle bucket of three consecutive vacant buckets, so both neighbors of the filled bucket are vacant at that moment. (Note that an extreme bucket, having no neighbor on one side, is never the middle of three buckets and is thus never filled.) Demotions cannot create a pair of adjacent occupied buckets. Since no two adjacent buckets are occupied before the search, the property is preserved by every operation.

\emph{Claim 2.} Fix a level $\ell$ and any moment during the execution, consider a maximal run of consecutive vacant buckets at level $\ell$ at that moment. By the invariants holding before the search, every maximal run of vacant buckets was then of length at most two. By the structural fact above, every bucket outside $W_\ell$ that bounds such a run remains occupied throughout. Hence, the run under consideration is contained in the union of at most two vacant buckets immediately to the left of $W_\ell$, the window $W_\ell$ itself, and at most two vacant buckets immediately to the right of $W_\ell$, which gives at most $2+3+2=7$ buckets in total.
\end{proof}
\section{Working Set}\label{s:ws}

Let $X = \langle x_1, x_2, \dots, x_m\rangle $ be a search sequence and let $f: [n] \to [n]$ be a uniformly random permutation. Let $X' = f(X) = \langle x'_1, x'_2, \dots, x'_m \rangle$, where $x'_i = f(x_i)$. For any $i \in \nf$ we define, analogously to $\mrs_i(r)$, a version that uses the transformed search sequence
\[
\mrs'_i(r) \coloneqq \max\{j \in [i-1] \mid x'_j \in [x'_i - r, x'_i + r]\}.
\]
As $f$ is a bijection, an item is searched for the first time at the same positions in $X$ and $X'$. So the set $\nf$ is the same for both sequences and $\mrs'_i(r)$ is well-defined for $i \in \nf$ since the previous search of $x_i$ lies in every window.

For $i \in \nf$, let $W_i \coloneqq \{x_j \mid \mrs_i \le j \le i\}$ be the \emph{working set} of $x_i$ at time $i$, which is the set of distinct items searched from the previous search of $x_i$ through time $i$. Let $w(x_i) \coloneqq |W_i|$ be the \emph{working set number} of $x_i$ at time $i$ (in this section $w$ refers to working set and not weight). Note that $x_i \in W_i$, so $w(x_i) \geq 1$. Like $\nf$, the sets $W_i$ and the numbers $w(x_i)$ depend only on the access pattern and not on the key values, and thus are the same for $X$ and $X'$.

Applying the definition of $OPT$, \eqref{e:opt}, to the permuted sequence $X'$ gives:
\[
OPT(f(X)) = \sum_{i \in \nf}
\left(
\log \min\left\{r \geq 1 \,\middle|\,
x'_{\mrs'_i(\frac{n}{r})} = x'_i\right\}
+
\log
\frac{n}{\min(x'_i,\, n-x'_i+1)}
\right)
\]

\begin{theorem}\label{t:ws}
Let $n \geq 2$ and let $f: [n] \to [n]$ be a uniformly random permutation. Then, for any access sequence $X$ of length $m \geq 2n$,
\[
\mathbb{E}_f[OPT(f(X))] = \Theta\left(m + \sum_{i \in \nf} \log w(x_i)\right).
\]
\end{theorem}

\begin{proof}
For $i \in \nf$, let
\[
M_i^f \coloneqq \min\left\{ r \geq 1 \,\middle|\, x'_{\mrs'_i(\frac{n}{r})} = x'_i \right\}.
\]
Note that $M_i^f$ is finite and for $r > n$ the window $[x'_i - \frac{n}{r}, x'_i + \frac{n}{r}]$ has radius less than $1$ and so contains no key other than $x'_i$, hence $\mrs'_i(\frac{n}{r}) = \mrs_i$ and the condition holds. By linearity of expectation,
\[
\mathbb{E}_f[OPT(f(X))] = \sum_{i \in \nf} \mathbb{E}\left[ \log M_i^f \right] + \sum_{i \in \nf} \mathbb{E}\left[ \log \frac{n}{\min(x'_i, n - x'_i + 1)} \right].
\]
We evaluate each expectation separately.

\paragraph{Expectation of the second term:}
As $f$ is a uniformly random permutation, the key $x'_i = f(x_i)$ is uniformly distributed over $[n]$, meaning $\mathbb{P}(x'_i = k) = \frac{1}{n}$ for all $k \in [n]$. So we have
\[
\mathbb{E}\left[ \log \frac{n}{\min(x'_i, n-x'_i+1)} \right] = \frac{1}{n} \sum_{k=1}^{n} \log \frac{n}{\min(k, n-k+1)}
\]
We bound this summation from both sides by positive constants. For the lower bound, every $k \in [n]$ satisfies $\min(k, n-k+1) \leq \frac{n+1}{2}$, so each of the $n$ terms is at least $\log \frac{2n}{n+1} \geq \log \frac{4}{3}$, using $n \geq 2$. For the upper bound, each value $v \in \left[1, \lceil \frac{n}{2} \rceil\right]$ occurs as $\min(k, n-k+1)$ for at most two values of $k \in [n]$, so
\[
\frac{1}{n} \sum_{k=1}^{n} \log \frac{n}{\min(k, n-k+1)} \leq \frac{2}{n} \sum_{v=1}^{n} \log \frac{n}{v} = \frac{2}{n} \log \frac{n^n}{n!} \leq 2 \log e,
\]
where the last inequality uses $n! \geq \left(\frac{n}{e}\right)^n$. Therefore the expectation is bounded above and below by positive constants:
\[
\mathbb{E}\left[ \log \frac{n}{\min(x'_i, n-x'_i+1)} \right] = \Theta(1)
\]
Summing over $i \in \nf$ yields a cost of $\Theta(|\nf|)$. As each item can be searched for the first time only once, $|\nf| \geq m - n$, and since $m \geq 2n$ we have $\frac{m}{2} \leq |\nf| \leq m$, the cumulative cost of the second term is therefore $\Theta(m)$.

\paragraph{Expectation of the first term:}

Fix $i \in \nf$. If $w(x_i) = 1$, then no item other than $x_i$ is searched between times $\mrs_i$ and $i$, so $x'_{\mrs'_i(\frac{n}{r})} = x'_i$ for every $r \geq 1$. Hence $M_i^f = 1$ and $\mathbb{E}[\log M_i^f] = 0 = \log w(x_i)$. Assume for the remainder that $w(x_i) \geq 2$.

We first reduce the failure event to a geometric one. For any radius $d \geq 0$, the search at time $\mrs_i$ is to $x_i$, whose key $x'_i$ lies in every window $[x'_i - d, x'_i + d]$, so $\mrs'_i(d) \geq \mrs_i$. Furthermore, by the definition of $\mrs_i$, no search at time in $(\mrs_i, i)$ is to $x_i$. Therefore
\[
x'_{\mrs'_i(d)} \neq x'_i \iff \text{some item of } W_i \setminus \{x_i\} \text{ is mapped by } f \text{ into } [x'_i - d, x'_i + d].
\]
In particular, the set $\{r \geq 1 \mid x'_{\mrs'_i(\frac{n}{r})} = x'_i\}$ is upward closed. As $r$ grows the window shrinks, and an item of $W_i \setminus \{x_i\}$ lying inside the smaller window also lies inside the larger one. Hence
\[
2^\ell < M_i^f \implies x'_{\mrs'_i(\frac{n}{2^\ell})} \neq x'_i \implies 2^\ell \leq M_i^f.
\]

Next we discretize $r$ using powers of two. For any real $M \geq 1$, the number of integers $\ell \geq 0$ with $2^\ell < M$ is $\lceil \log_2 M \rceil$, and the number with $2^\ell \leq M$ is $\lfloor \log_2 M \rfloor + 1$. Writing $I$ for an indicator variable and combining with the implications above,
\[
\log M_i^f \;\leq\; \sum_{\ell = 0}^{\infty} I\left[ x'_{\mrs'_i(\frac{n}{2^\ell})} \neq x'_i \right] \;\leq\; \log M_i^f + 1.
\]
The sum is finite, as for $2^\ell > n$ the window has radius less than $1$ and the indicator is $0$. Taking expectations and applying linearity of expectation,
\[
\mathbb{E}\left[ \log M_i^f \right] = \sum_{\ell = 0}^{\infty} P_\ell - \theta_i \text{ for some } \theta_i \in [0,1], \qquad \text{where } P_\ell \coloneqq \mathbb{P}\left[ x'_{\mrs'_i(\frac{n}{2^\ell})} \neq x'_i \right].
\]
By the reduction above, $P_\ell$ is the probability that at least one of the $w(x_i) - 1$ items of $W_i \setminus \{x_i\}$ is mapped within distance $\frac{n}{2^\ell}$ of $x'_i$. Conditioned on the value of $x'_i$, these items are assigned distinct keys uniformly at random, without replacement, from the $n - 1$ keys other than $x'_i$.

We now divide the summation at the threshold $\ell^* = \lfloor \log_2 w(x_i) \rfloor$. Note $\ell^* \geq 1$ as $w(x_i) \geq 2$:

$$
\sum_{\ell=0}^\infty P_\ell = \sum_{\ell=0}^{\ell^*} P_\ell + \sum_{\ell = \ell^* + 1}^\infty P_\ell
$$

We bound each of the summations:
\begin{enumerate}
    \item When $\ell \le \ell^*$: We claim $\sum_{\ell=0}^{\ell^*} P_\ell = \Theta(\log w(x_i))$. The upper bound is immediate as $P_\ell \leq 1$ and $\log_2 w(x_i) \geq 1$,
    \[
    \sum_{\ell=0}^{\ell^*} P_\ell \leq \ell^* + 1 \leq \log_2 w(x_i) + 1 \leq 2 \log_2 w(x_i).
    \]
    For the lower bound, first observe that $P_0 = 1$, the window of radius $n$ contains all of $[n]$, and $W_i \setminus \{x_i\}$ is nonempty. Now suppose $\ell \leq \ell^* - 2$, so that $2^{\ell+2} \leq w(x_i)$. The window $[x'_i - \frac{n}{2^\ell}, x'_i + \frac{n}{2^\ell}] \cap [n]$ contains at least $\min(\lfloor n/2^\ell \rfloor, n-1) \geq \frac{n}{2^{\ell+1}}$ keys other than $x'_i$ (note $2^\ell \leq w(x_i) \leq n$), so a single item of $W_i \setminus \{x_i\}$ lands in the window with probability at least $\frac{n/2^{\ell+1}}{n-1} \geq \frac{1}{2^{\ell+1}}$. As the items are placed without replacement, conditioning on earlier items missing the window only increases the fraction of remaining keys that lie inside it. The probability that all $w(x_i) - 1$ items miss is therefore at most the product of the individual miss probabilities:
    \[
    \left(1 - \frac{1}{2^{\ell+1}}\right)^{w(x_i)-1} \leq e^{-\frac{w(x_i)-1}{2^{\ell+1}}} \leq e^{-1},
    \]
    where the last inequality uses $w(x_i) - 1 \geq \frac{w(x_i)}{2} \geq 2^{\ell+1}$. Hence $P_\ell \geq 1 - e^{-1}$ for all $\ell \leq \ell^* - 2$, and
    \[
    \sum_{\ell=0}^{\ell^*} P_\ell \geq \max\left\{1,\ (1 - e^{-1})(\ell^* - 1)\right\} \geq \frac{1 - e^{-1}}{3} \log_2 w(x_i),
    \]
    where the last inequality uses $\log_2 w(x_i) \leq \ell^* + 1 \leq 3 \max\{1, \ell^* - 1\}$.

    \item When $\ell > \ell^*$: The window contains at most $\frac{2n}{2^\ell}$ keys other than $x'_i$, so a single item of $W_i \setminus \{x_i\}$ lands in the window with probability at most $\frac{2n/2^\ell}{n-1} \leq \frac{4}{2^\ell}$ (using $n \geq 2$, which follows from $w(x_i) \geq 2$). By the union bound,
    \[
    P_\ell \leq \frac{4(w(x_i) - 1)}{2^\ell},
    \]
    and summing the geometric series,
    \[
    \sum_{\ell = \ell^* + 1}^\infty P_\ell \leq \sum_{\ell = \ell^* + 1}^\infty \frac{4 w(x_i)}{2^\ell} = \frac{4 w(x_i)}{2^{\ell^*}} < 8 = O(1),
    \]
    using $w(x_i) < 2^{\ell^* + 1}$.
\end{enumerate}

Combining both ranges with the discretization error $\theta_i \in [0,1]$, we have, for $w(x_i) \geq 2$,
\[
\mathbb{E}\left[ \log M_i^f \right] = \Theta(\log w(x_i)) + O(1),
\]
and since $\log_2 w(x_i) \geq 1$, there are positive constants $c_1, C, C'$, independent of $i$, such that
\[
c_1 \log w(x_i) - C \;\leq\; \mathbb{E}\left[ \log M_i^f \right] \;\leq\; C' \log w(x_i).
\]
These bounds also hold when $w(x_i) = 1$, as then $\mathbb{E}[\log M_i^f] = 0 = \log w(x_i)$.

\paragraph{Combining the two terms:}
Let $S \coloneqq \sum_{i \in \nf} \log w(x_i)$. The upper bound is immediate:
\[
\mathbb{E}_f[OPT(f(X))] = \sum_{i \in \nf} \mathbb{E}\left[ \log M_i^f \right] + \Theta(m) \leq C' S + O(m) = O\left(m + S\right).
\]
For the lower bound, both terms of \eqref{e:opt} are nonnegative. The first by the restriction $r \geq 1$ and the second as $\min(x'_i, n - x'_i + 1) \leq n$, so $\mathbb{E}_f[OPT(f(X))]$ is at least each of the two sums separately. The second sum gives, for some constant $c_2 > 0$,
\[
\mathbb{E}_f[OPT(f(X))] \geq c_2 |\nf| \geq \frac{c_2}{2} m,
\]
and the first sum gives
\[
\mathbb{E}_f[OPT(f(X))] \geq \sum_{i \in \nf} \mathbb{E}\left[ \log M_i^f \right] \geq c_1 S - C |\nf| \geq c_1 S - C m.
\]
If $c_1 S \leq 2 C m$, then $S = O(m)$ and the first bound alone is $\Omega(m + S)$. Otherwise, $c_1 S - C m \geq \frac{c_1}{2} S$, and the maximum of the two bounds is at least their average, which is $\Omega(m + S)$. In either case,
\[
\mathbb{E}_f[OPT(f(X))] = \Omega\left(m + \sum_{i \in \nf} \log w(x_i)\right),
\]
which together with the upper bound completes the proof.
\end{proof}

\section*{Acknowledgments}

AI tools were used to critique proofs, identify bugs, mistakes, and inconsistencies. Specifically, Claude Fable and Opus were high-quality proofreaders that undoubtedly improved the paper's quality and were particularly attentive to ensuring edge cases were handled correctly.

For the two transformations presented in sections~\ref{s:popbtree} and~\ref{s:poptree}, given the tedious and non-innovative nature of these reductions, AI was used more aggressively and was given the sequence of lemmas in the manner that one would explain them to a colleague to flush out. AI then produced first drafts of the proofs quickly which were then human edited and verified. We suspect that most humans, after looking at Figure~\ref{fig:skippoplist-to-bst}, would intuitively understand the transformation from skip list to tree and find it conceivable that a constant balance factor can be maintained at every step of the way; the two appendices confirm this intuition.

In this way we have separated the paper into the entirely human-written body, proofread by AI, and the AI appendix, proofread and edited by humans. The authors feel such a mixed approach where the human-written novel ideas are separated from working out the easily believable but not so interesting is one possible paradigm as we try to work out what papers should look like in this new age.

The figure was also generated interactively with Claude over many rounds; its tikz skills are much better than the authors'.

\bibliographystyle{alpha}
\bibliography{bib}

\appendix
\section{PopBTree} \label{s:popbtree}

We now describe the way to interpret a SkipPopList as a multiway search
tree. The construction is the standard skip-list to multiway search tree of~\cite{DBLP:conf/ACMse/DeanJ07} but
specialized to our invariants.

\paragraph{Top items.}
We define the \emph{level} of an item $x$, written $\ell(x)$, as the smallest
level of the SkipPopList in which $x$ is stored. Equivalently, it is the
topmost level in the SkipPopList in which $x$ resides. We call the items whose top level is exactly $\ell$,
the \emph{top items of level $\ell$}, and write
$$
\mathit{Top}_\ell \;=\; \{\, x \in [n] \;:\; \ell(x)=\ell \,\}.
$$
Since every key is stored in
$\ell_{\max}$ and each key has a well-defined
smallest level, the sets $\mathit{Top}_1,\mathit{Top}_2,\dots,\mathit{Top}_{\ell_{\max}}$ partition $[n]$.
Moreover, the items stored at level $\ell$ are exactly
$\bigcup_{\ell' \le \ell} \mathit{Top}_{\ell'}$

\paragraph{Sentinels and virtual buckets.}
Throughout, we use two \emph{sentinels}, $0$ and $n+1$, whose top level we
define to be $0$. They are not stored in the SkipPopList and occupy no real
bucket. They are a notational device, made precise as follows. Extend each
level $\ell$ by two \emph{virtual buckets}: a bucket $-1$, permanently
occupied by the sentinel $0$, and a bucket $3^\ell$, permanently
occupied by the sentinel $n+1$. These buckets contain no keys.

Virtual bucket $-1$ of level $\ell$ is the rightmost child of virtual bucket $-1$ of level $\ell -1$, consistent with the indexing $3\cdot(-1)+2=-1$, and virtual
bucket $3^\ell$ is the leftmost child of virtual bucket $3^{\ell-1}$, as
$3\cdot 3^{\ell-1}+0=3^\ell$. Under this convention, the boundary invariants of SkipPopList
become instances of the interior. The extreme buckets being vacant says
exactly that no occupied bucket, virtual or real, is adjacent to another at
the boundary, and Invariant~\ref{i:empty} says that between any two
consecutive occupied buckets, sentinels included, there are one or two
vacant buckets. We may therefore treat sentinel endpoints exactly like stored
items in the bucket in the counting arguments below.

\paragraph{The tree $\mathcal{T}$.}
The tree $\mathcal{T}$, which we call the \emph{PopBTree}, is a rooted multiway
search tree whose keys are the $n$ items of $[n]$. Each node $N$ in the PopBTree is associated
with an open key-interval $(a,b)$, its \emph{gap}. The root's gap is
$(0,n+1)$. If a node has gap $(a',b')$ and keys $x_1<\dots<x_k$, the gaps of
its children are the intervals between consecutive elements of the sequence
$a',x_1,\dots,x_k,b'$. An endpoint of a gap is thus either a key of the parent
or inherited from the parent's own gap. Unwinding the recursion, each of $a$
and $b$ is a key of some proper ancestor of $N$ or a sentinel, and in general
the two endpoints belong to two \emph{different} ancestors. The node $N$ with interval
$(a,b)$ is built as follows:

\begin{itemize}
    \item Among the items lying strictly inside $(a,b)$, let $\ell^\star$ be the
    smallest top level attained, i.e.\ the highest level of the SkipPopList
    reached by any item of $(a,b)$. The \emph{keys} of $N$ are the top items of
    level $\ell^\star$ that lie in $(a,b)$, listed in increasing order
    $x_1 < x_2 < \dots < x_k$. We say $N$ is a node \emph{at level
    $\ell^\star$}.

    \item The keys of $N$ split $(a,b)$ into the $k+1$ subintervals
    $(a,x_1),(x_1,x_2),\dots,(x_{k-1},x_k),(x_k,b)$. Each subinterval that
    contains at least one item is covered by a child node of $N$, defined by
    the same rule. A subinterval containing no items contributes no child.
\end{itemize}

The following basic facts about gaps address a subtlety worth stressing: the
endpoints of a node's gap are in general \emph{not} keys of its parent, so
statements about them must be routed through the ancestors that contribute
them.

\begin{lemma}{\bf Gap endpoints. }\label{l:endpoints}
Let $N$ be a node of $\mathcal{T}$ at level $\ell$ with gap $(a,b)$. Then
\begin{enumerate}
    \item[(i)] $\ell(a)<\ell$ and $\ell(b)<\ell$.
    \item[(ii)] Every item strictly inside $(a,b)$ has top level at least
    $\ell$, and the items of top level exactly $\ell$ inside $(a,b)$ are
    precisely the keys of $N$.
    \item[(iii)] If $N$ is not the root, then the level of $N$ is strictly
    greater than the level of its parent. Consequently, levels increase
    strictly along every root-to-node path, and each endpoint of $(a,b)$ has
    top level at most the level of the parent of $N$.
\end{enumerate}
\end{lemma}

\begin{proof}
We proceed by induction along the path from the root to $N$. For the root, the endpoints of its gap are the sentinels, whose top level is
$0$. Since every node has positive level, (i) holds. Statement (ii) is exactly
the definition of the root's keys.

Now suppose the claims hold for a node $N'$ of level $\ell'$ with gap
$(a',b')$, and let $N$ be the child corresponding to a gap $(a,b)$ of $N'$. Each endpoint $a$ or $b$ is either a key of $N'$ or an endpoint of
$(a',b')$. In the first case its top level is $\ell'$. In the second case, the
induction hypothesis gives top level less than $\ell'$. Thus $\ell(a),\ell(b)\le \ell'$.

Next consider any item strictly inside $(a,b)$. Since
$(a,b)\subseteq(a',b')$, the induction hypothesis implies that its top level is
at least $\ell'$. Moreover, it cannot have top level exactly $\ell'$, because
such an item would be a key of $N'$, and the keys of $N'$ are not contained
strictly inside one of its gaps. Therefore every item inside $(a,b)$ has top
level at least $\ell'+1$.

By definition, the level of $N$ is the minimum top level among the items inside
$(a,b)$, hence $\ell\ge \ell'+1$, this proves (iii). Combining this with $\ell(a),\ell(b)\le\ell'$ gives
$\ell(a),\ell(b)<\ell$, which is exactly (i). Finally, (ii) follows directly from the definition of the keys of $N$: they
are precisely the items inside $(a,b)$ whose top level equals the minimum top
level $\ell$.
\end{proof}

\paragraph{Why this is a search tree.}
The construction is exactly the branching structure of a $B$-tree. Every node
$N$ covering $(a,b)$ with keys $x_1<\dots<x_k$ has its children living in the
gaps $(a,x_1),\dots,(x_k,b)$ between consecutive keys. The keys of $N$
separate the key-ranges of its children in the usual $B$-tree fashion. A search
descending into $N$ compares the query against $x_1,\dots,x_k$ and follows the
unique gap that contains it. Iterating from the root, every item $x$ is a key of
exactly one node, the node covering the gap in which $x$ attains the highest
level. So $\mathcal{T}$ stores each of the $n$ items exactly once and is a valid
search tree over $[n]$.

Equivalently, and this is the phrasing we will use, the nodes of $\mathcal{T}$
are obtained level by level. At level $\ell$, delete from the items stored at
level $\ell$ all those already stored at level $\ell-1$. The remaining items are exactly $\mathit{Top}_\ell$, and each maximal set of them
that is consecutive among the items stored at level $\ell$ forms one node at level $\ell$.

We now bound the number of keys stored at any node $N$ in the PopBTree using the above construction and the invariants of the SkipPopList.

\begin{lemma}{\bf Node size. }\label{l:nodesize}
Every node of $\mathcal{T}$ has at most $4$ keys and hence at most $5$ children.

\end{lemma}

\begin{proof}

Let $N$ be a node at level $\ell$ covering the gap $(a,b)$. If $N$ is the root, then $\ell=1$. The only non-extreme bucket at level $1$ is
the middle bucket, so the root has exactly one key. Thus the claim holds.

Now suppose $N$ is not the root. By Lemma~\ref{l:endpoints}(i), both endpoints
of the gap satisfy $\ell(a),\ell(b)<\ell$. Therefore each endpoint occupies a
(possibly virtual) bucket at level $\ell-1$. If the endpoint is a stored item,
it appears in the bucket containing its key by the subset property. If it is a
sentinel, it occupies its corresponding virtual bucket. These two buckets are consecutive. Indeed, if there were another occupied
bucket between them at level $\ell-1$, it would contain an item whose top level
is at most $\ell-1$. That item would lie strictly inside $(a,b)$, contradicting
Lemma~\ref{l:endpoints}(ii), which states that every item inside $(a,b)$ has top
level at least $\ell$.

Now, consider the occupied buckets at level $\ell$ whose keys lie strictly
between $a$ and $b$. Their stored items have top level exactly $\ell$, since an
item stored at level $\ell-1$ would also occupy a bucket between $a$ and $b$ at
level $\ell-1$, contradicting the fact that the endpoint buckets are
consecutive. By Lemma~\ref{l:endpoints}(ii), these items are exactly the keys of
$N$. It remains to count these buckets.

\emph{Spacing at level $\ell-1$:}
Let $a$ occupy bucket $p$ and $b$ bucket $q$ at level $\ell-1$. As $a,b$ are
consecutive occupied buckets, no two adjacent buckets are occupied and at most
two consecutive buckets are vacant (Invariant~\ref{i:empty}), so the number of
vacant buckets between them is $1$ or $2$ which gives that $q-p\in\{2,3\}$.

\emph{Passing to level $\ell$:}
Each level-$(\ell-1)$ bucket has three children at level $\ell$. Indexing the buckets
from $0$, bucket $p$ has children $3p,3p+1,3p+2$ and bucket $q$ has children
$3q,3q+1,3q+2$. Item $a$ occupies the child of $p$ that contains its key, some
bucket $\alpha=3p+r_a$ with $r_a\in\{0,1,2\}$, and likewise $b$ occupies
$\beta=3q+r_b$. (For a sentinel endpoint this holds by the virtual-bucket
convention, with $r_a=2$ for the sentinel $0$ and $r_b=0$ for the sentinel
$n+1$.) Every level-$\ell$ bucket with index strictly between $\alpha$ and
$\beta$ has all of its keys in $(a,b)$, while buckets $\alpha$ and $\beta$
themselves hold only $a$ and $b$. Hence the keys of $N$ occupy buckets in
$\{\alpha+1,\dots,\beta-1\}$, and
$$
\beta-\alpha \;=\; 3(q-p)+(r_b-r_a)\;\le\; 3\cdot 3 + 2 \;=\; 11 .
$$

\emph{Counting occupied buckets:}
Bucket $\alpha$ is occupied by $a$, so by the no-two-adjacent-occupied invariant,
bucket $\alpha+1$ is vacant. Symmetrically $\beta-1$ is vacant. Every key of $N$
therefore occupies a bucket in $[\alpha+2,\beta-2]$, a run of
$\beta-\alpha-3\le 8$ consecutive buckets in which no two adjacent buckets may be
occupied. This gives at most $\lceil 8/2\rceil = 4$ occupied buckets, so
$N$ has at most $4$ keys and, one child per non-empty gap between consecutive
keys giving, at most $5$ children.

\end{proof}

Next we define the \emph{depth} of a node in $\mathcal{T}$ to be the number of nodes on the
path from the root to that node, inclusive.
Thus the root has depth~$1$.

\begin{lemma}{\bf Depth. }\label{l:depth}
An item whose top level in the SkipPopList is $\ell$ is a key of a node at depth
$\ell$ in $\mathcal{T}$. Equivalently, locating such an item from the root visits
exactly $\ell$ nodes.
\end{lemma}

\begin{proof}
By construction the keys of a node are top items of a common level, and an item
$x$ with top level $\ell$ is a key of exactly one node which lies at
level $\ell$. It therefore suffices to prove that a node at level $\ell$ has depth
$\ell$, which we do by showing that the transform never skips a level, i.e., the parent
of every non-root node at level $\ell$ is a node at level $\ell-1$.

Let $N$ be a non-root node at level $\ell$ covering the gap $(a,b)$, and let
$N'$ be its parent, a node at level $\ell'$. By Lemma~\ref{l:endpoints}(iii),
$\ell\ge\ell'+1$ and $\ell(a),\ell(b)\le\ell'$, it remains to show
$\ell\le\ell'+1$.

Each of $a,b$ occupies a (possibly virtual) bucket at level
$\ell'$, at positions $p<q$ with $q\ge p+2$, since no two adjacent buckets are
occupied. Bucket $p+1$ is therefore vacant at level $\ell'$, being adjacent to
the occupied bucket $p$, and every key of bucket $p+1$ lies strictly between
$a$ and $b$ because it exceeds every key of bucket $p$, hence exceeds $a$, and
precedes every key of bucket $q$, hence precedes $b$ (if $p$ or $q$ is
virtual the corresponding comparison against the sentinel is trivial). As
$\ell'<\ell_{\max}$, Lemma~\ref{l:below} applies to bucket $p+1$ which has an
occupied child at level $\ell'+1$, whose occupant has top level exactly
$\ell'+1$ and key inside bucket $p+1$, hence inside $(a,b)$. So the smallest
top level attained inside $(a,b)$ is at most $\ell'+1$, giving
$\ell\le\ell'+1$. Therefore $\ell=\ell'+1$, i.e.\ $N'$ is at level $\ell-1$.

The claim now follows by induction on $\ell$. The root is the unique node at
level $1$ and has depth $1$. For $\ell\ge 2$, a node at level $\ell$ has a parent
at level $\ell-1$, which has depth $\ell-1$ by induction, so the node has depth
$\ell$. Hence an item of top level $\ell$, being a key of a level-$\ell$ node,
sits at depth $\ell$.
\end{proof}

In Lemma~\ref{l:nodesize} we bounded the degree of the PopBTree. We now show that in the PopBTree all leaves are at the same depth.

\begin{lemma}{\bf Uniform depth. }\label{l:uniform}
Every node of $\mathcal{T}$ at depth less than $\ell_{\max}$ has at least one
child, so every leaf of $\mathcal{T}$ is at depth $\ell_{\max}$.
\end{lemma}

\begin{proof}
Let $N$ be a node at level $\ell<\ell_{\max}$, covering $(a,b)$, and let $x_1$ be
its leftmost key, an item of top level $\ell$, so $x_1$ occupies a bucket at
level $\ell$. To its left within $(a,b)$ there is a vacant bucket. Indeed, by
Lemma~\ref{l:endpoints} and the subset property, $a$ occupies a (possibly
virtual) bucket $p$ at level $\ell$, and since no two adjacent occupied buckets
are permitted, sentinels included, by the virtual-bucket convention, $x_1$'s
bucket has index at least $p+2$.

Let $c$ be the bucket immediately left of
$x_1$, then $c$ is real (its index is at least $p+1\ge 0$), vacant (adjacent
to the occupied bucket of $x_1$), and its keys exceed $a$ and are below $x_1$,
hence lie in $(a,b)$. As $\ell<\ell_{\max}$, Lemma~\ref{l:below}
gives $c$ an occupied child at level $\ell+1$, whose key lies in $(a,b)$ and has
top level $\ell+1$. The node formed in the gap of $N$ containing that key is a
child of $N$. Hence $N$ is not a leaf, and every leaf lies at depth $\ell_{\max}$.
\end{proof}

Together, Lemmas~\ref{l:nodesize} and~\ref{l:uniform} say that $\mathcal{T}$ is a
$B$-tree because every node has between $1$ and $4$ keys and at most $5$ children, and
all root-to-leaf paths have the same length. We use both facts below.

Just as $BB[\alpha]$ tree is a BST in which each child's subtree carries at least an $\alpha$ fraction of the weight of its parent's subtree, we
now give a similar definition for $B$-trees which we call a $BBB[\alpha]$ tree.

Fix $\alpha\in(0,\tfrac12]$. For a node $v$ of a multiway search tree, let its
\emph{weight} $w(v)$ be one plus the number of keys stored in the subtree rooted
at $v$. Equivalently, $w(v)$ is the number of external nodes (empty child slots)
of that subtree. When every child slot of a node is nonempty, its external
nodes are partitioned among the subtrees of its children, and the weights of
the children then sum to the weight of the node. In general (empty slots
being external nodes of the parent itself) the children's weights sum to at
most the weight of the node.
\begin{definition}{\bf $BBB[\alpha]$ tree. }\label{d:bbb}

A \emph{$BBB[\alpha]$ tree} is a $B$-tree in which every
node has a bounded number of children and all leaves lie at the same depth and, for every internal node $v$ and every child $c$ of $v$,
$$
w(c)\;\ge\;\alpha\, w(v).
$$
\end{definition}

\begin{lemma}{\bf Subtree size. }\label{l:subtreesize}
Let $N$ be a node of $\mathcal{T}$ at level $\ell$, and let $s(N)$ be the number of items stored in the subtree rooted at $N$. Then
$$
\frac{n}{3^{\ell-1}} \;\le\; s(N) \;\le\; \frac{4n}{3^{\ell-1}} ,
$$
so $s(N)=\Theta\!\left(n/3^{\ell}\right)$, with the two bounds differing by a factor
of $4$.
\end{lemma}

\begin{proof}
The subtree rooted at $N$ stores exactly the data items whose key lies in the
interval $(a,b)$ that $N$ covers. Indeed, every key in $(a,b)$ has top level at
least $\ell$ and is a key of a unique node descended from $N$, while keys outside
$[a,b]$ lie outside $N$'s subtree and $a,b$ themselves are keys of ancestors of
$N$ or sentinels. Hence $s(N)=b-a-1$, the number of integers strictly between
$a$ and $b$. (For
the root, $a=0$ and $b=n+1$, giving $s(N)=n$, consistent with the bounds at
$\ell=1$.)

For a non-root node, $a$ and $b$ occupy consecutive (possibly
virtual) buckets at level $\ell-1$, exactly as in the proof of
Lemma~\ref{l:nodesize}, at positions $p<q$ with $q-p\in\{2,3\}$.
Each real bucket at level $\ell-1$ holds $n/3^{\ell-1}$ consecutive keys while
virtual buckets hold none. The interval
$(a,b)$ contains all of the $q-p-1\ge 1$ full buckets $p+1,\dots,q-1$, we have
$$
s(N)\;\ge\;(q-p-1)\,\frac{n}{3^{\ell-1}}\;\ge\;\frac{n}{3^{\ell-1}},
$$
and $(a,b)$ is contained in the union of the key ranges of at most
$q-p+1\le 4$ buckets $p,\dots,q$ (virtual buckets among these contribute no
keys, so the bound only improves), and we have
$$
s(N)\;\le\;(q-p+1)\,\frac{n}{3^{\ell-1}}\;\le\;\frac{4n}{3^{\ell-1}}.
$$
\end{proof}

We now use the above lemma to show that the PopBTree is $BBB[\alpha]$ balanced.

\begin{lemma}{\bf $\mathcal{T}$ is a $BBB[\alpha]$ tree. }\label{l:bbb}
$\mathcal{T}$ is a $BBB[\tfrac{1}{12}]$ tree.
\end{lemma}

\begin{proof}
By Lemma~\ref{l:nodesize} every node has at most $5$ children, and by
Lemma~\ref{l:uniform} all leaves lie at the common depth $\ell_{\max}$, so
$\mathcal{T}$ is a $B$-tree. It remains to check the weight condition of
Definition~\ref{d:bbb} with $\alpha=\tfrac1{12}$.

Let $N$ be an internal node at level $\ell$ and let $c$ be a child of $N$. Since
no level is skipped (Lemma~\ref{l:depth}), $c$ is at level $\ell+1$. Applying
Lemma~\ref{l:subtreesize}, the lower bound to $s(c)$ and the upper bound to
$s(N)$, gives
$$
w(c)=s(c)+1\ \ge\ \frac{n}{3^{\ell}}+1 ,
\qquad
w(N)=s(N)+1\ \le\ \frac{4n}{3^{\ell-1}}+1=\frac{12n}{3^{\ell}}+1 .
$$
Writing $x=n/3^{\ell}>0$,
$$
\frac{w(c)}{w(N)}\ \ge\ \frac{x+1}{12x+1}\ \ge\ \frac{1}{12},
$$
the last step because $12(x+1)=12x+12\ge 12x+1$. Thus $w(c)\ge\tfrac1{12}\,w(N)$
for every node $N$ and every child $c$, and $\mathcal{T}$ satisfies
Definition~\ref{d:bbb} with $\alpha=\tfrac1{12}$.
\end{proof}

\paragraph{Weak bounds during restoration.}
The lemmas so far describe $\mathcal{T}$ when all SkipPopList invariants
hold between searches. During the restoration that follows a search,
the structural invariants are transiently violated, but by
Lemma~\ref{l:weakinv}, after every unit-cost operation and at every level
$\ell<\ell_{\max}$, no two adjacent buckets are occupied, and at most seven
consecutive buckets are vacant, and level $\ell_{\max}$ never changes at all
(Lemma~\ref{l:bottom}). Moreover, the extreme buckets remain vacant at every
moment because demotions only vacate buckets, the promotion of $x_i$ stops below the
extreme buckets by definition, and a gap repair fills only the middle of
three consecutive vacant buckets. The virtual-bucket convention therefore remains in
force. Since the construction of $\mathcal{T}$ and Lemma~\ref{l:endpoints}
use only the subset property, $\mathcal{T}$ is well defined at every instant.

\begin{lemma}{\bf Weak structural bounds. }\label{l:weakstruct}
At every moment during the execution of the restoration method, that is,
after every unit-cost operation, the tree $\mathcal{T}$ read off from the
current SkipPopList contents satisfies:
\begin{enumerate}
    \item Every level other than level $1$ stores at least one item, so the
    root of $\mathcal{T}$ is at level $1$ or $2$.
    \item Every non-root node has at most $12$ keys, the root has at most
    $4$ keys, and every node has at most $13$ children.
    \item The level of a child exceeds the level of its parent by at most two.
    \item Every node $N$ at level $\ell$ has subtree size
    $n/3^{\ell-1} \le s(N) \le 9n/3^{\ell-1}$.
    \item Every gap of every node at level $\ell<\ell_{\max}$ contains at
    least one key, so every non-leaf node has exactly one more child than it
    has keys.
    \item For every node $N$ and every child $c$ of $N$,
    $w(c)\ge\tfrac{1}{81}\,w(N)$.
\end{enumerate}
\end{lemma}

\begin{proof}
(1) A level $\ell\ge 2$ storing no item would be a run of $3^\ell\ge 9>7$
consecutive vacant buckets, contradicting Lemma~\ref{l:weakinv} (for
$\ell=\ell_{\max}$, Lemma~\ref{l:bottom}). Hence the smallest top level
attained in $(0,n+1)$ is $1$ or $2$.

(2) Repeat the count of Lemma~\ref{l:nodesize} with the weak spacing. The
endpoints of a non-root node's gap occupy consecutive occupied (possibly
virtual) buckets at the level above the node. The argument there uses only
Lemma~\ref{l:endpoints} and the subset property so the vacant run between
them now has length at most $7$, hence $q-p\le 8$, giving
$\beta-\alpha\le 3\cdot 8+2=26$ and at most $\lceil 23/2\rceil=12$ occupied
buckets in $[\alpha+2,\beta-2]$.

The root's keys are the occupied
buckets of its level, which is $1$ or $2$ by (1). Level $1$ has one
non-extreme bucket and level $2$ has seven, of which at most
$\lceil 7/2 \rceil=4$ can be occupied, no two being adjacent.

(3) Let $(u,v)$ be a nonempty gap of a node $N$ at level $\ell$. As in (2),
$u$ and $v$ occupy consecutive occupied buckets at level $\ell$, at positions
$s<t$ with $t-s\le 8$, and at level $\ell+1$ they occupy buckets
$\sigma=3s+r_u<\tau=3t+r_v$. The $\tau-\sigma-1\ge 3(t-s)-3\ge 3$ buckets
strictly between $\sigma$ and $\tau$ have all their keys in $(u,v)$. If one
of them is occupied, its occupant lies in $(u,v)$ and has top level exactly
$\ell+1$ because top level at most $\ell$ would contradict
Lemma~\ref{l:endpoints}(ii). The gap's node is then at level $\ell+1$. If all
are vacant which is impossible when $\ell+1=\ell_{\max}$, then by
Lemma~\ref{l:bottom} at most two consecutive buckets are vacant, so assume
$\ell+2\le\ell_{\max}$. Then, at least $3\cdot 3=9$ children of these buckets are at level $\ell+2$, likewise
have all their keys in $(u,v)$, and $9>7$ consecutive vacant buckets are
forbidden (by Lemma~\ref{l:weakinv}, or by Lemma~\ref{l:bottom} if
$\ell+2=\ell_{\max}$), so one of them is occupied. Its occupant has top level
exactly $\ell+2$, its parent bucket at level $\ell+1$ being vacant. In either
case the smallest top level in $(u,v)$ is $\ell+1$ or $\ell+2$.

(4) Both endpoints of $N$'s gap have top level below $\ell$
(Lemma~\ref{l:endpoints}) and hence occupy (possibly virtual) buckets at
level $\ell-1$ directly, whatever the level of $N$'s parent. As in (2) these
are consecutive occupied buckets with $q-p\in[2,8]$. The interval $(a,b)$
contains $q-p-1\ge1$ real buckets between them and is contained in
the key ranges of $q-p+1\le 9$ buckets $p,\dots,q$, so
$n/3^{\ell-1}\le s(N)\le 9n/3^{\ell-1}$.

(5) A gap of a node at level $\ell<\ell_{\max}$ is delimited by occupants of
buckets at level $\ell$ at distance at least $2$. The real bucket between
them holds $n/3^{\ell}\ge 1$ keys (as $\ell\le\log_3 n$), all strictly inside
the gap.

(6) Let $c$ be a child of $N$, with $N$ at level $\ell$. By (3), $c$ is at
level $\ell+1$ or $\ell+2$, so by (4), $w(c)=s(c)+1\ge n/3^{\ell+1}+1$, while
$w(N)\le 9n/3^{\ell-1}+1=27\,n/3^{\ell}+1$. Writing $x=n/3^{\ell}$,
$$
81\left(\frac{x}{3}+1\right)=27x+81\ \ge\ 27x+1,
$$
so $w(c)\ge\tfrac1{81}\,w(N)$.
\end{proof}

In particular $\mathcal{T}$ satisfies the weight condition of
Definition~\ref{d:bbb} with $\alpha=\tfrac1{81}$ at \emph{every} moment of
the execution, not merely between searches. The one property that can fail
transiently is uniform leaf depth. By (3) a child may sit two levels below
its parent, so the leaves, all of which are at level $\ell_{\max}$, may
transiently lie at depths between $\lceil\ell_{\max}/2\rceil$ and
$\ell_{\max}$. We return to this point when routing the pointer during
restoration.

During the restoration the pointer must visit a scattered set of nodes for the
clearing, promotion, and repair of nodes using only parent and child moves.
All the nodes lie near the search path in \emph{bucket} space, and we show
that the subtree connecting them to the root has total size
$O(\ell_i(x_i))$. A depth-first traversal of the nodes then visits every node
cheaply. Crucially, the size bound rests only on the weak structural bounds
of Lemma~\ref{l:weakstruct}, so it holds at every moment of the restoration.

\begin{lemma}{\bf Connecting subtree. }\label{l:tour}
Let $c\ge 2$ be a constant. Fix a search $x_i$, let
$\ell_i=\ell_i(x_i)$ be its maximum level at the start of the search, and for
each level $\lambda$ let $B_\lambda$ be the bucket containing the key $x_i$
at level $\lambda$. Let the \emph{window} $W_\lambda$ be the union of the
key ranges of the buckets at distance at most $c$ from $B_\lambda$. At any
moment during the search or its restoration, let $T_S$ consist of the root
together with every node whose level $\lambda$ satisfies
$\lambda\le\ell_i+2$ and whose gap intersects $W_\lambda$. Then $T_S$ is a
subtree of $\mathcal{T}$ containing the root, with $O(1)$ nodes at each
level and $|T_S|=O(\ell_i)$ in total.
\end{lemma}

\begin{proof}
Bucket $B_{\lambda-1}$ is the parent of $B_\lambda$,
and each level-$(\lambda-1)$ bucket spans three level-$\lambda$ buckets, so
the buckets at distance at most $c$ from $B_\lambda$ are contained in the
buckets at distance at most $\lceil (c+2)/3\rceil\le c$ from $B_{\lambda-1}$
(using $c\ge 2$), hence $W_\lambda\subseteq W_{\lambda-1}$.

Now, let $N\in T_S$ be a non-root node at level
$\lambda$ with parent $N'$ at level $\lambda'<\lambda$. The gap of $N'$
contains the gap of $N$, which intersects
$W_\lambda\subseteq W_{\lambda'}$, so $N'\in T_S$. Being closed under taking
parents and containing the root, $T_S$ is a subtree.

Finally, distinct nodes at level $\lambda$ have
disjoint gaps, and by Lemma~\ref{l:weakstruct}(4), which holds at every
moment and each such gap contains at least $n/3^{\lambda-1}$ keys, i.e.\ spans
at least three level-$\lambda$ bucket lengths in key space. The window
$W_\lambda$ spans $2c+1$ bucket lengths, so at most $(2c+1)/3+2=O(1)$ gaps
of level-$\lambda$ nodes can intersect it. Summing over at most
$\ell_i+2$ levels gives $|T_S|=O(\ell_i)$.
\end{proof}

A depth-first traversal of $T_S$ from the root, descending only into
children that belong to $T_S$, moves the pointer along each edge of $T_S$ at
most twice and hence visits every node of $T_S$ in $O(\ell_i)$ pointer
moves,
only unit-cost operations, not computation.
The bound holds when the
traversal is interleaved with \emph{edits}, the constant-size updates of
$\mathcal{T}$ defined before Lemma~\ref{l:localupdate} below, provided each
edit modifies
only the pointer's current node, its children, or its ancestors on the
current traversal stack, and replaces $O(1)$ nodes by $O(1)$ nodes whose
gaps lie in the same windows. The traversal simply continues on the current
tree. Over the whole restoration only $O(\ell_i)$ nodes
ever belong to $T_S$ ($O(\ell_i)$ initially, plus $O(1)$ per promotion or
demotion), so the total number of pointer moves remains $O(\ell_i)$.

We maintain $\mathcal{T}$ so that it always mirrors the current SkipPopList. This means that the
keys of the level-$\ell$ nodes are exactly the top items of level $\ell$. A single
promotion or demotion of the model changes the top level of one item by one, and
we show that the matching update of $\mathcal{T}$ can be implemented by rewriting the keys and child pointers of $O(1)$ nodes, all within one pointer hop of the pointer; we call such a constant-size update an \emph{edit}.

\begin{lemma}{\bf Constant-time promotion and demotion. }\label{l:localupdate}
Suppose that, at any moment during a search or its restoration, the pointer
rests on the node $N$ of $\mathcal{T}$ that stores an item $x$ as a key, the
node at $x$'s current top level. Then the promotion or demotion of $x$, as
defined in the skip-list cost model, can be mirrored in $\mathcal{T}$ in
$O(1)$ time, touching only nodes within one pointer hop of $N$.
\end{lemma}

\begin{proof}
Let $x$ have top level $\ell$. Every node mentioned below has $O(1)$ keys and
children by Lemma~\ref{l:weakstruct}(2), which holds at every moment of the
restoration. We do not assume the full invariants.

\emph{Demotion:} ($\ell\to\ell+1$; possible only for $\ell<\ell_{\max}$).
Within $N$, the key $x$ is flanked by the two gaps $(u,x)$ and $(x,v)$, where
$u$ and $v$ are the keys of $N$ neighboring $x$ or, if $x$ is the leftmost or
rightmost key, the corresponding endpoint of $N$'s gap (possibly a sentinel).
Removing $x$ from level $\ell$ merges these two gaps into the single gap
$(u,v)$. Its node is a new node $c$. Since $x$ is stored at level $\ell+1$ by
the subset property now has top level $\ell+1$, and no other item of
$(u,v)$ has top level below $\ell+1$. So node $c$ is at level $\ell+1$ and
has $x$ among its keys. Concretely, $c$ is assembled from $O(1)$ existing
pieces.

The child of $N$ in the gap $(u,x)$ exists, the gap containing at
least one item by Lemma~\ref{l:weakstruct}(5), as $\ell<\ell_{\max}$. If
that child is at level $\ell+1$, its keys and child subtrees form the
portion of $c$ left of $x$. If it is at level $\ell+2$ (a transient level
skip, Lemma~\ref{l:weakstruct}(3)), it instead becomes the child of $c$
hanging in the gap immediately left of $x$. Symmetrically on the right. Finally $x$ is deleted from $N$. If $x$ was the
only key of $N$, then $N$ ceases to exist and $c$ takes $N$'s place as a
child of $N$'s parent. If $N$ was the root, $c$ becomes the root, consistent
with Lemma~\ref{l:weakstruct}(1). In every case a constant number of nodes,
each within one pointer hop of $N$, are read and rewritten. Moving the
pointer down onto $c$ finishes the operation, matching the cost model, in
which a demotion moves the pointer down one level.

\emph{Promotion:} ($\ell\to\ell-1$; possible only for $\ell>1$, and performed
by the maintenance method only after the destination bucket and its neighbors
have been cleared). Here $M:=N$ stores $x$ at level $\ell$. Inserting $x$ into the level
above turns $x$ into a delimiter, splitting $M$ at $x$. The keys of $M$ left
of $x$, with their child subtrees, form $M_L$, and those right of $x$ form
$M_R$. If $M_L$ (respectively $M_R$) has no keys, it is not a node, and its
place is taken directly by the leftmost (rightmost) child subtree of $M$, one
level further down, or by nothing if that side of $M$'s gap holds no items.
If $M$ is the root, it has no parent, and the promotion is into level $1$.
Note that the root not being at level $1$ forces it to be at level $2$ with level $1$
empty (Lemma~\ref{l:weakstruct}(1)). So after the promotion, $x$ is the
unique item of top level $1$, and $x$ becomes the new root, a single-key
node at level $1$ with $M_L$ and $M_R$ as its children. Only $M$ and the
$O(1)$ pointers among its pieces are touched.

Otherwise, let $N'$ be the
parent of $M$.
What happens above then depends on the level of $N'$. If $N'$ is at level
$\ell-1$, then $x$, with new top level $\ell-1$, lies strictly inside $N'$'s
gap and becomes a key of $N'$. The child slot held by $M$ is replaced by the
key $x$ flanked by two slots holding $M_L$ and $M_R$. If instead $N'$ is at
level $\ell-2$ (a transient skip, as arises during gap repair), then $x$
becomes a new single-key node at level $\ell-1$ occupying $M$'s slot as a
child of $N'$, with $M_L$ and $M_R$ as its children. Only $M$, $N'$, and the
$O(1)$ pointers between them are touched. Moving the pointer up onto the node
now containing $x$ finishes the operation, matching the cost model.

In both directions a constant number of nodes, each of constant size, are
split or merged, and only a constant number of keys and child pointers are
rewritten, all within one pointer hop of where the pointer began. Hence each
operation runs in $O(1)$ time.
\end{proof}

That the nodes stay of constant size throughout a maintenance pass is not a
bookkeeping accident but an instance of Lemma~\ref{l:weakstruct}(2): after
every unit-cost operation every node has at most $12$ keys, and the bound of
$4$ from Lemma~\ref{l:nodesize} is restored together with the invariants.

Next, to run the maintenance with a pointer that cannot move between siblings, we
need to know that all its operation sites lie in the connecting subtree.

\begin{lemma}{\textbf{Locality.  }}\label{l:locality}
There is an absolute constant $c_0$ such that, for every search, every
unit-cost operation of the restoration method acts on an item whose bucket,
at the level of the operation, is within distance $c_0$ of the bucket
containing the searched key at that level, and every operation takes place
at a level at most $\ell_i(x_i)$. Consequently, at every moment, every node
of $\mathcal{T}$ at which an operation must be performed, and every node
created, destroyed, or modified, belongs to the connecting subtree $T_S$ of
Lemma~\ref{l:tour} with constant $c=\max(c_0,2)$.
\end{lemma}

\begin{proof}
First the buckets, read off the maintenance method. A clearing cascade
demotes an item $y$ that occupies, at some level $\lambda<\ell_i(x_i)$, the
bucket containing $x_i$ or a bucket adjacent to it. The cascade acts at the
levels from $y$'s top level down to $\lambda$. At any such level $\lambda'$,
the bucket of $y$ is an ancestor of its bucket at level $\lambda$, and since
ancestors of adjacent buckets are equal or adjacent, the bucket of $y$ at
level $\lambda'$ is within distance one of the bucket of $x_i$'s key there.
The promotion of $x_i$ acts on the buckets containing $x_i$
itself---distance zero. A gap repair fills the middle of three consecutive
vacant buckets, at least one of which was vacated by a demotion and is
therefore within one of $x_i$'s bucket, so the filled bucket is within four
of $x_i$'s bucket. The item promoted into it comes from a child of the
filled bucket, within $3\cdot 4+2=14$ of $x_i$'s bucket one level down.
Taking $c_0=14$ covers every case.

Now pass to $\mathcal{T}$. An operation on an item whose bucket is within
$c_0$ of $B_\lambda$ is performed, per Lemma~\ref{l:localupdate}, at the
node holding that item as a key. That node's gap contains the item's key and
hence intersects the window $W_\lambda$, so the node is in $T_S$ (its level
is at most $\ell_i+1\le \ell_i+2$). The $O(1)$ nodes rewritten by the
operation, the node itself, its parent, and the replacements assembled from
its children, have gaps meeting the window at their own levels, all at most
$\ell_i+2$, and so also lie in $T_S$.
\end{proof}

The algorithm may start from any SkipPopList satisfying the invariants; the initial PopBTree, and as we later describe the PopTree, are read off from it.

\begin{theorem}\label{t:tree-opt}
{\bf PopBTree. }
There is an online algorithm that maintains the PopBTree $\mathcal{T}$ over
$[n]$, using a single pointer that moves only from a node to its parent or to
one of its children, such that: after every search $\mathcal{T}$ is a
$BBB[\tfrac1{12}]$ tree, after every unit-cost operation $\mathcal{T}$
satisfies the weight condition of Definition~\ref{d:bbb} with
$\alpha=\tfrac1{81}$, and any search sequence $X$ of length $\Omega(n\log n)$
is processed in total time $O(OPT(X))$.
\end{theorem}

\begin{proof}
The algorithm keeps $\mathcal{T}$ equal to the tree of the current SkipPopList,
by Lemma~\ref{l:bbb} this is a $BBB[\tfrac1{12}]$ tree between searches.
Consider
the $i$th search, for $x_i$, and let $\ell_i=\ell_i(x_i)$ be its top level just
before the search. We bound the work on this search by $O(\ell_i)$.

\emph{Search.} Starting at the root, repeatedly compare $x_i$ against the $O(1)$
keys of the current node and descend into the child whose interval contains $x_i$,
until $x_i$ is found. By Lemma~\ref{l:depth} the node storing $x_i$ sits at depth
$\ell_i$, so this is $\ell_i$ parent-to-child moves with $O(1)$ work each,
i.e.\ $O(\ell_i)$ time, and uses no sibling moves.

\emph{Restoration.} The algorithm then mirrors in $\mathcal{T}$ the
SkipPopList maintenance that reinstates the invariants after searching
$x_i$. That maintenance performs $O(\ell_i)$ promotions and demotions in the
prescribed order, clearing demotions, then the promotion of $x_i$, then gap
repairs, as tallied in the proof of Lemma~\ref{l:skipcost}. By
Lemma~\ref{l:locality}, every operation site
lies at every moment in the connecting subtree $T_S$ of Lemma~\ref{l:tour}.
The pointer is routed in three passes, one per phase, respecting the order
of the method.

\emph{Clearing pass.} A depth-first traversal of $T_S$. Whenever the pointer
reaches a node holding an item $y$ that must be cleared, with the node's
level being $y$'s top level, it executes $y$'s entire cascade: each demotion
is an $O(1)$-time edit at the pointer (Lemma~\ref{l:localupdate}) that drops
the pointer one level, onto the node into which $y$ was merged, precisely
where the next demotion of $y$ takes place, since a demotion is permitted
only at an item's top level. When the cascade ends the pointer climbs back
up and resumes the traversal. The climb is paid for by the cascade's own
operations. Because the traversal visits a node before its descendants, and
merges performed by demotions only combine nodes below the pointer, every
cascade top is reached and every cascade runs top-down as the model
requires.

\emph{Promotion pass.} The pointer returns to the root and descends to the
node holding $x_i$, of depth at most its level $\ell_i$
(Lemma~\ref{l:weakstruct}(3)), then promotes $x_i$ one level at a time, the
pointer riding up with each promotion (Lemma~\ref{l:localupdate}), until the
bucket above is extreme or level $1$ is reached.

\emph{Repair pass.} A second depth-first traversal of $T_S$. At each repair
site the pointer stands on the node, one level below the vacancy, holding
the item to be promoted, and performs the promotion in $O(1)$ time
(Lemma~\ref{l:localupdate}). The edit touches only the pointer's node and
its parent, which is on the traversal stack. Repairs do not cascade, as a
promotion vacates nothing.

Each pass moves the pointer $O(\ell_i)$ times: the two traversals run over a
subtree that has $O(\ell_i)$ nodes at every moment (Lemma~\ref{l:tour}) and
is edited only at the pointer, its children, and its stacked ancestors
(Lemma~\ref{l:localupdate}), and the promotion pass is a single descent and
ascent. Hence the restoration costs $O(\ell_i)$. Once it finishes the
invariants hold (Lemma~\ref{l:maintain}), so by Lemma~\ref{l:bbb}
$\mathcal{T}$ is again a $BBB[\tfrac1{12}]$ tree, and after every
intermediate unit-cost operation $\mathcal{T}$ satisfies the weight
condition with $\alpha=\tfrac1{81}$, by Lemma~\ref{l:weakstruct}(6).

 Summing, the algorithm spends $O\!\left(\sum_i \ell_i(x_i)\right)$
time, which by Corollary~\ref{cor:depth} is $O(OPT(X))$ whenever
$|X|=\Omega(n\log n)$.
\end{proof}

\section{PopTree}\label{s:poptree}

The PopBTree described in the previous subsection, although it achieves optimal cost for any search sequence $X$, is not a BST. In this section we describe how we turn it into a binary search tree, the
\emph{PopTree}, that for any search sequence $X$ has cost $O(OPT(X))$.

To convert the PopBTree into PopTree, we expand each PopBTree node into a right-leaning chain, and
show four things:
\begin{enumerate}
    \item The expansion preserves depth up to a constant factor.

    \item Each unit of PopBTree work becomes $O(1)$ unit-cost operations of the BST
model, i.e., moving  a single pointer to the parent, left child, or right child,
and single rotations of the pointed node with its parent.

\item The PopTree is a $BB[\tfrac1{20}]$ tree whenever the SkipPopList invariants hold.

 \item The PopTree is a $BB[\alpha_1]$ tree, for an absolute constant $\alpha_1>0$, after every unit-cost operation.
\end{enumerate}

\paragraph{The right-chain expansion.}
Let $N$ be a PopBTree node with keys $k_1<k_2<\dots<k_m$ (so $1\le m\le 4$ by
Lemma~\ref{l:nodesize}) and child subtrees $c_0,c_1,\dots,c_m$, where $c_t$ holds
the keys between $k_t$ and $k_{t+1}$, with $c_0$ the keys below $k_1$ and $c_m$
those above $k_m$. Replace $N$ by the \emph{right chain} on its keys: $k_1$ at the
top, and $k_{t+1}$ the right child of $k_t$ for $1\le t<m$. Hang the (recursively
converted) child subtrees on this chain by making $c_{t-1}$ the left child of
$k_t$ for $1\le t\le m$, and $c_m$ the right child of $k_m$. The result is a binary
search tree on the keys of $N$ and its children: from $k_t$ the right subtree holds
keys larger than $k_t$ and the left subtree holds the interval $(k_{t-1},k_t)$.
Applying this expansion at every node turns the PopBTree $\mathcal{T}$ into a binary
search tree, the \emph{PopTree}, denoted as $\mathcal{P}$ (see Figure \ref{fig:skippoplist-to-bst}).

\begin{lemma}{\bf Depth is preserved. }\label{l:poptree-depth}
An item at level $\ell$ of the PopBTree lies at depth between $\ell$ and $4\ell$
in the PopTree. In particular an item at PopBTree depth $\ell_i(x_i)$ is at PopTree
depth $\Theta(\ell_i(x_i))$.
\end{lemma}

\begin{proof}
The root-to-item path in $\mathcal{P}$ passes through the chains of the $\ell$
PopBTree nodes on the item's root-to-node path. It enters each chain at its top and
walks down at most $m\le 4$ chain nodes before stepping into a child (or reaching
the item), so its length is at most $4\ell$ and at least $\ell$, one node per chain.
\end{proof}

\begin{lemma}{\bf Constant-time simulation by rotations. }\label{l:poptree-sim}
Each unit-cost PopBTree operation,
is realized on the PopTree by $O(1)$ unit-cost operations of the BST model.
The PopTree after the $O(1)$ unit-cost operations is exactly the
right-chain expansion of the updated PopBTree.
\end{lemma}

\begin{proof}
\emph{Pointer moves:} Each PopBTree node is a chain of $O(1)$ PopTree nodes,
at every moment (Lemma~\ref{l:weakstruct}(2)). To move to a child of $N$,
descend within $N$'s chain to the relevant key and step to its left or right
child. To move to the parent of $N$, ascend to the top of $N$'s chain and
take one step further, both $O(1)$ moves, and no rotations.

\emph{Promotion:} Let $x=k_j$ be a key of the chain of $M$, whose keys are
$k_1<\dots<k_m$ with child subtrees $c_0,\dots,c_m$, and let the pointer rest
on $x$. Rotate $x$ with its parent, $j-1$ times. We claim that afterwards
$x$ occupies the position where the top of $M$'s chain stood, with left
subtree exactly the expansion of $M_L$ (the keys $k_1,\dots,k_{j-1}$ with
children $c_0,\dots,c_{j-1}$) and right subtree exactly the expansion of
$M_R$ (the keys $k_{j+1},\dots,k_m$ with children $c_j,\dots,c_m$). The
right subtree of $x$ is the chain $k_{j+1},\dots,k_m$ and is never touched.
For the left, induct on the rotations: when $x$ is rotated above $k_t$ (for
$t=j-1,\dots,1$), the node $k_t$ becomes the left child of $x$, keeps its
left subtree $\exp(c_{t-1})$, and receives as right subtree the previous
left subtree of $x$, which by induction is the expansion of the block
$k_{t+1},\dots,k_{j-1}$ with children $c_t,\dots,c_{j-1}$. This results in the
expansion of the block $k_t,\dots,k_{j-1}$ with children
$c_{t-1},\dots,c_{j-1}$. At $t=1$ this is $\exp(M_L)$.

It remains to move $x$ one level up. If $M$ hung as the left child of a
chain node $v$ of its parent's chain, one further rotation of $x$ with $v$
makes $x$ the predecessor of $v$ in that chain, with left subtree
$\exp(M_L)$, while $v$ receives $\exp(M_R)$ as its left subtree. This is exactly the
expansion of the PopBTree after the promotion, in which $x$ is inserted as a
key before $v$ with $M_L,M_R$ in the two flanking slots. If $M$ hung as the
right child of the last chain node $u$ of its parent's chain, no further
rotation is needed as $x$ is already the right child of $u$ and extending the
chain by the key $x$ with subtrees $\exp(M_L),\exp(M_R)$ gives the correct
expansion. In total at most $m\le O(1)$ rotations, each of the pointed node
with its parent, with the pointer riding on $x$ throughout. The degenerate
cases of Lemma~\ref{l:localupdate} when $M_L$ or $M_R$ empty, the promotion
creating a new single-key node one level up during a transient skip, or $M$
being the root, this correspond to shorter instances of the same rotation trace. In the root case no further rotation is needed, as $x$ already stands at the
root position after the $j-1$ rotations, with subtrees $\exp(M_L)$ and
$\exp(M_R)$.

\emph{Demotion:} This is the inverse rotation sequence. The pointer rests on
$x$, a chain node whose PopBTree demotion merges the two blocks flanking
$x$. Repeatedly, the pointer steps to the child of $x$ that must rise, the
top of the flanking block's chain and rotates it up, then steps back to
$x$. Each such iteration is one child move, one rotation, and one move back, and
undoes one step of the promotion trace read backwards. After $O(1)$
iterations the tree is the expansion of the PopBTree after the demotion, and
the pointer stands on $x$, now a chain node of the merged node one level
down, matching the cost model. If the demotion removes the last key of $x$'s
old node, nothing special occurs in the PopTree because there are no empty chains,
only the absence of one, and the merged chain hangs where the old chain
hung.
\end{proof}

We next show that the PopTree is a $BB[\alpha]$ tree.

\begin{lemma}{\bf The PopTree is $BB[\alpha]$. }\label{l:poptree-bb}
The PopTree $\mathcal{P}$ is a $BB[\tfrac1{20}]$ tree.
\end{lemma}

\begin{proof}
Give each node $u$ of $\mathcal{P}$ the weight $w(u)$. We show that each child subtree of $u$  has weight at least
$\tfrac1{20}\,w(u)$, which is the $BB[\tfrac1{20}]$ condition.

Let $u=k_j$ be a PopTree node lying in the chain of a PopBTree node $N$ at level
$\ell$, with keys $k_1,\dots,k_m$ and children $c_0,\dots,c_m$. The subtree of $u$
consists of the keys $k_j,\dots,k_m$ together with the child subtrees
$c_{j-1},\dots,c_m$, whose external nodes are disjoint, so
$w(u)=\sum_{t=j-1}^{m} w(c_t)$,
a sum of $m-j+2\le 5$ terms. The left child subtree of $u$ is the converted
$c_{j-1}$, of weight $w(c_{j-1})$, and the right child subtree is $k_{j+1},\dots,k_m$
with $c_j,\dots,c_m$ when $j<m$ (weight $\sum_{t=j}^m w(c_t)$), and the converted
$c_m$ when $j=m$.

\emph{Internal node $N$ ($\ell<\ell_{\max}$):} Every gap of an internal node has a
nonempty child. Its two delimiters occupy consecutive occupied buckets at
level $\ell$, sentinels included, by the virtual-bucket convention, which
covers the outer gaps of the root and the leftmost and rightmost
nodes, and are separated by a vacant bucket, which by Lemma~\ref{l:below}
has an occupied child one level down whose occupant lies in the gap. Hence
all $c_0,\dots,c_m$ are nonempty, each a level-$(\ell+1)$ subtree
(Lemma~\ref{l:depth}), so by Lemma~\ref{l:subtreesize},
$$
S+1\ \le\ w(c_t)\ \le\ 4S+1 ,\qquad S:=\frac{n}{3^{\ell}}\ \ge\ 1 .
$$
Then $w(u)\le(m-j+2)(4S+1)\le 5(4S+1)=20S+5$, while each child subtree contains at
least one $c_t$ and so has weight at least $S+1$. Therefore each child subtree has
weight ratio
$$
\frac{S+1}{20S+5}\ \ge\ \frac1{20}.
$$

\emph{Leaf node $N$ ($\ell=\ell_{\max}$).} Here $u$'s subtree contains only the $\le 4$
keys $k_j,\dots,k_m$, so $w(u)\le 5$, while every child subtree has weight at least
$1$ so $1\ge\tfrac1{20}\cdot 5$, so the condition holds here too.

Thus every node of $\mathcal{P}$ has both child subtrees of weight at least
$\tfrac1{20}\,w(u)$, so $\mathcal{P}$ is a $BB[\tfrac1{20}]$ tree.
\end{proof}

In a binary tree the weights of the two children of a node sum to the
weight of the node, where missing children are counted with weight one. Hence the lower
bound $w(\text{child})\ge\alpha\,w(u)$ for both children implies the upper
bound $w(\text{child})\le(1-\alpha)\,w(u)$. Thus Lemma~\ref{l:poptree-bb}
says precisely that the PopTree satisfies Definition~\ref{d:balanced} with
$\alpha=\tfrac1{20}$ and $\beta=\tfrac{19}{20}$.

\begin{lemma}{\bf Balance at every step. }\label{l:poptree-always}
There is an absolute constant $\alpha_1>0$, our crude accounting gives
$\alpha_1=\tfrac1{729}$, such that after every unit-cost operation of the
PopTree algorithm, including every single rotation inside the rotation
sequences of Lemma~\ref{l:poptree-sim}, the PopTree is a $BB[\alpha_1]$
tree.
\end{lemma}

\begin{proof}
Pointer moves do not change the tree, so it suffices to consider the trees
occurring between consecutive rotations. By Lemma~\ref{l:poptree-sim} each
such tree is the right-chain expansion of the current PopBTree in which at
most one $O(1)$-size region is part-way through an edit, and by the
inductive description in that proof, every BST node of the region is a chain
node of a block.

Let $u$ be any BST node at such a moment, a chain node of a block of a
PopBTree node at some level $\ell$, and write $x=n/3^\ell$. The subtree of
$u$ is contained in the subtree of the topmost PopBTree node involved, whose
level is at least $\ell-2$ (an edit spans at most the levels
$\ell-2,\dots,\ell+2$, by Lemma~\ref{l:weakstruct}(3) and the descriptions
in Lemma~\ref{l:localupdate}), so by Lemma~\ref{l:weakstruct}(4),
$$
w(u)\ \le\ \frac{9n}{3^{\ell-3}}+1\ =\ 243x+1 .
$$
Each of the two subtrees of $u$ either contains a full child subtree of a
PopBTree node of the region, of level at most $\ell+2$, and then by
Lemma~\ref{l:weakstruct}(4) has weight at least $n/3^{\ell+1}+1=x/3+1$. Otherwise
it contains no PopBTree child subtree at all, which by
Lemma~\ref{l:weakstruct}(5) forces the block, and $u$'s whole subtree, to
lie at leaf level, where every weight is bounded by an absolute
constant (a merged leaf block has $O(1)$ keys by
Lemma~\ref{l:weakstruct}(2)), comfortably below $729$, so that the
external or key-only subtree, of weight at least $1$, meets the ratio. In
the former case
$$
\frac{x/3+1}{243x+1}\ \ge\ \frac{1}{729},
$$
since $729(x/3+1)=243x+729\ge 243x+1$.

\end{proof}

We now present our main theorem.

\begin{theorem}{\bf PopTree.}\label{t:poptree}
There is an online algorithm maintaining a binary search tree over $[n]$, the
PopTree, that uses only the unit-cost operations of the $BB[\alpha]$-algorithm
model
such that the tree is a
$BB[\tfrac1{729}]$ tree after every unit-cost operation, is a $BB[\tfrac1{20}]$ tree after
every search, and processes any search sequence $X$ of length
$\Omega(n\log n)$ in total time $O(OPT(X))$.
\end{theorem}

\begin{proof}
Run the PopBTree algorithm of Theorem~\ref{t:tree-opt} on the right-chain
expansion. By Lemma~\ref{l:poptree-sim} each of its pointer moves,
promotions, and demotions is realized by $O(1)$ unit-cost operations of the
BST model
and carries only a
constant-factor overhead. Hence the work for the $i$th search, which is
$O(\ell_i(x_i))$ on the PopBTree (Theorem~\ref{t:tree-opt}), remains
$O(\ell_i(x_i))$ on the PopTree, consistent with Lemma~\ref{l:poptree-depth},
whereby the search descends to PopTree depth $O(\ell_i(x_i))$. Summing over the
sequence and applying Corollary~\ref{cor:depth},
$$
\sum_i O(\ell_i(x_i))\ =\ O(OPT(X))
\qquad\text{for } |X|=\Omega(n\log n).
$$

 After every unit-cost operation the PopTree is a
$BB[\tfrac1{729}]$ tree by Lemma~\ref{l:poptree-always}, as the model demands of
a valid $BB[\tfrac1{729}]$-algorithm. Between searches the SkipPopList
invariants hold, so the PopBTree is a $BBB[\tfrac1{12}]$ tree by
Lemma~\ref{l:bbb} and its expansion is a $BB[\tfrac1{20}]$ tree by
Lemma~\ref{l:poptree-bb}. Thus the whole sequence is executed in
$O(OPT(X))$ time by an algorithm that is at all times a valid
$BB[\tfrac1{729}]$-algorithm.
\end{proof}

\end{document}